\RequirePackage{fix-cm}
\documentclass[10pt]{svjour3}

\usepackage[margin=1.3in]{geometry}

\smartqed

\usepackage{amsbsy}
 \usepackage{amsmath}
\usepackage{amssymb}
\usepackage{graphicx}
\usepackage{url}
\usepackage{mathrsfs}      
\usepackage{subfigure}
\usepackage{amsmath}
\usepackage{euscript}
\usepackage{xcolor}
\usepackage{algorithm}
\usepackage[noend]{algorithmic}

\usepackage{epstopdf}
\usepackage{siunitx}

\newtheorem{thm}{Theorem}
\newtheorem{lem}{Lemma}

\DeclareMathOperator*{\argmin}{argmin}

\newcommand{\bb}[1]{\mathbb{#1}}
\renewcommand{\v}[1]{\boldsymbol{#1}}
\newcommand{\m}[1]{\mathrm{#1}}
\renewcommand{\c}[1]{\mathcal{#1}}
\newcommand{\idef}{\stackrel{\mathrm{def}}{=}}
\newcommand{\var}{\mathbb{V}\mathrm{ar}}

\begin{document}

\title{Accurate Computation of 
the Distribution of  Sums of Dependent Log-Normals with
Applications to the Black-Scholes Model}

\subtitle{Sums of Dependent Log-Normals}
\author{Zdravko~I.~Botev \and Robert Salomone\and Daniel Mackinlay}
				
				\institute{School of Mathematics and Statistics, The University of New South Wales,
        Sydney, NSW 2052, Australia
           \and
          Department of Mathematics, The University of Queensland, Brisbane, QLD 4072,  Australia\and School of Mathematics and Statistics, The University of New South Wales,
        Sydney, NSW 2052, Australia
}

\maketitle
\begin{abstract}
We present a new Monte Carlo methodology  for the accurate estimation of the distribution of the sum of dependent log-normal random variables. The methodology delivers  statistically unbiased estimators for three distributional quantities of significant  interest in finance and risk management:
the left tail, or cumulative distribution function; the probability density function; and the right tail, or complementary distribution function of the sum of dependent log-normal factors.  In all of these three cases our methodology delivers fast and highly accurate estimators in settings for which  existing methodology delivers estimators with large variance that tend to underestimate the true quantity of interest. We provide insight into  the computational challenges using  theory and numerical experiments, and explain their much wider implications  for Monte Carlo statistical estimators of  rare-event probabilities.    
In particular, we find that theoretically strongly-efficient  estimators 
should  be used with great caution in practice,
 because they may yield  inaccurate results  in the pre-limit. Further, this inaccuracy 
 may not be detectable from the output of the Monte Carlo simulation, because 
 the simulation output  may severely underestimate the
 true variance of the estimator. 

\end{abstract}

\keywords{Log-Normal distribution \and Rare event simulation \and 
logarithmic efficiency\and  large deviations\and Conditional Monte Carlo\and
Quasi Monte Carlo\and second-order efficiency}

\section{Introduction}
The distribution of the sum of log-normals (SLN) has numerous practical applications \cite{limpert2001log} ---
in the pricing of  Asian options under a Black-Scholes  model
 \cite{bacry2013log,dufresne2004log,milevsky1998asian}; in wireless systems in telecommunications
\cite{gubner2006new,rached2016unified} in insurance  value-at-risk computations \cite{embrechts2014academic,mcneil2015quantitative,zuanetti2006lognormal}; and recently even in the modelling of viral
 social media phenomena \cite{doerr2013lognormal}. For this reason, the accurate computation of characteristics of the SLN distribution  are receiving increasing attention.

The first 
left-tail efficient Monte Carlo method for the estimation of the SLN cumulative distribution function (cdf) was proposed 
by Gulisashvili  and Tankov \cite{gulisashvili2016tail}. This was then followed
by Asmussen et al. \cite{asmussen2014laplace,asmussen2016exponential,laub2016approximating} who approximate the cdf using Laplace transform techniques. Up until these seminal works,
the only available  approximations of the cdf were deterministic  moment-matching heuristics, whose accuracy quickly deteriorates beyond very few dimensions (see \cite{asmussen2016exponential} for a survey of these).  

With the exception of the seminal work \cite{gulisashvili2016tail}, all of the existing proposals can only deal with the distribution of the sum of independent log-normals, or, in the case of
\cite{laub2016approximating}, with the Laplace transform of the SLN distribution. 

Other examples of research in the area include the efficient estimation of the right tail of the SLN distribution under the assumption of  independent log-normal factors  \cite{asmussen2006improved,nguyen2014new}, and, of more consequence for practical applications, under the assumption of correlated factors \cite{asmussen2011efficient,asmussen2008asymptotics,kortschak2013efficient,kortschak2014second}. 
 
In this article, we present  new Monte Carlo estimators  for the cdf (distribution function), pdf (density function), and right tail (complementary distribution function) of SLN distribution. 
Regarding these three new proposals, our original contributions can be summarized as follows.

\paragraph{Cdf estimator.} 
We  propose a new asymptotically efficient estimator of the cdf with superior practical performance than its nearest competitor. We also show that under some mild conditions, the estimator is not only asymptotically efficient, but strongly efficient with vanishing relative error. This means  that its accuracy becomes better and better as we move further and further into the tail. In addition, while existing methods estimate and study either the left or 
right tails of the SLN distribution, our estimator is the first to also
estimate efficiently the distribution and density in a non-asymptotic 
regime (that is, in the main body of the distribution).    


\paragraph{Pdf estimator.}  Our novel estimator of the pdf of the SLN distribution is infinitely smooth in the model parameters. As a result of this, in a Quasi Monte Carlo setting, 
this smoothness  accelerates the rate of convergence beyond 
that of the canonical Monte Carlo rate of $\c O(\sqrt{n})$  ($n$ is the 
Monte Carlo sample size). This Quasi Monte Carlo acceleration is peculiar 
to our proposal only --- all other competing estimators are either not smooth, or cannot deal with  dependency.

\paragraph{Right tail estimator.}   We show both numerically and theoretically that many of the existing proposals for estimating the right tail of the SLN distribution \cite{asmussen1997simulation,asmussen2006improved,asmussen2011efficient,gulisashvili2016tail,kortschak2013efficient} can be unreliable in some simple examples of applied interest. 
More precisely, while the existing estimators work satisfactory when the log-normal variates are independent, these estimators exhibit exploding variance in cases of positively correlated  log-normals. 
Unfortunately, dependence structures which induce strong positive correlation are 
precisely the cases of practical interest in finance and reliability (the computation of such tail probabilities arises, for example, in estimating the likelihood of a large loss from a portfolio with asset prices driven by the Black-Scholes geometric Brownian motion model \cite[Chapter 15]{kroese2011handbook}). 

 In addition to proving that our estimator is  asymptotically efficient as we move deeper and deeper into the right tail, we show that,  in
a number of practical settings,  it is  more accurate than its competitors by many orders of magnitude.

Further to this, we provide a refinement of the tail asymptotics of the lognormal distribution (item 1 of Lemma~\ref{lem}), and use this refinement to prove that our estimator is second-order efficient.
A second-order efficient estimator is one whose precision or standard error can be estimated reliably from simulation, a property only enjoyed by our new estimator (Corollary~\ref{cor:var}).

Finally, it is frequently the case that we not only wish to estimate the probability of a rare-event, but also wish to draw random states conditional on the rare-event. In this article  we propose the first  sampler
for simulating from the SLN distribution conditional on
a left-tailed rare event. A notable feature of our sampling algorithm is that the simulation is not approximate (as in Markov chain Monte Carlo), but exact.

\paragraph{Wider Implications}
One important conclusion
 with wider implications is that 
 a theoretically strongly-efficient  estimator may perform poorly in practice, and a weakly-efficient estimator might be preferable in the pre-limit. 
Our suggested solution to this problem 
is to use estimators with guaranteed second-order efficiency.  In this way,  any potential shortcomings of 
the  estimator are more likely to  become  evident during the course of the Monte Carlo simulation.

%

The rest of the paper is organized around the 
three qualitatively different parts of the SLN distribution: 
 (1)  the left tail of the SLN distribution;
(2) the density of the main body of the distribution; (3) the right tail of SLN distribution. 
In all three cases we wish to control the (quasi) Monte Carlo error of the estimator.

The left tail and main body is covered in Section~\ref{sec:cdf}, and
the right tail is considered in  Section~\ref{sec:right-tail}. In Section~\ref{sec:ISVE} we review
the \emph{importance sampling vanishing error} (ISVE) estimator proposed in
\cite{asmussen2011efficient}, and show numerically how in some cases it may yield highly inaccurate estimates that tend to severely underestimate the true probability. We give some intuitive explanations for the poor performance of the estimator and then in Section~\ref{sec:new}  describe our novel estimator and its theoretical properties. This is followed by   numerical illustrations of the main theoretical findings, and a concluding section.

\section{Cumulative Distribution and Density}
\label{sec:cdf}
We start by considering the cumulative  distribution function of the SLN:
\[
\ell(\gamma)=\bb P(X_1+\cdots+X_d\leq\gamma),
\]
where: (1) $\v X$ are (dependent) log-normal random variables governed by  a Gaussian copula, so that $\ln\v X\sim \mathsf{N}(\v\nu,\Sigma)$ for some positive definite covariance matrix $\Sigma$; and (2) the parameter  $\gamma>0$ is allowed to  be a  small enough threshold so that $\ell$ is a small or rare-event probability.

Then, if 
$\m L\m L^\top=\Sigma$ is the lower triangular decomposition of the covariance matrix, we can write
\[\textstyle
\ell=\bb P(\exp(\nu_1+\m L_{11}Z_1)+\cdots+\exp(\nu_d+\sum_{j}\m L_{dj}Z_j)\leq\gamma),
\]
where under the measure $\bb P$, we have $\v Z\sim\mathsf{N}(\v 0,\m I)$.
In other words, under $\bb P$ we can set 
$X_k=\exp\left(\nu_k+\sum_{i\leq k}\m L_{ki}Z_i\right)$, which we henceforth assume. To proceed, note that the events 
\[
\{X_1\leq \gamma\}\supseteq \{X_1+X_2\leq \gamma\}\supseteq\cdots\supseteq \{X_1+\cdots+X_d\leq \gamma\}
\]
are nested. In other words, if we define 
\[
\alpha_j(z_1,\ldots,z_{j-1})\idef
  \frac{\ln(\gamma-\sum_{k<j}x_k)-\nu_j-\sum_{k<j}\m L_{jk}z_k}{\m L_{jj}} ,\quad j > 1,
\]
then, the following events are nested:
\[
\{Z_1\leq \alpha_1\}\supseteq \{Z_2\leq \alpha_2(Z_1)\}\supseteq\cdots\supseteq \{Z_d\leq \alpha_d(Z_1,\ldots,Z_{d-1})\},
\]
with the last one being equivalent to the event of interest. 

Let $\mathsf{TN}_{(l,u)}(\mu,\sigma^2)$ denote the normal
distribution $\mathsf{N}(\mu,\sigma^2)$ truncated to the interval 
$(l,u)$. 
The above nested sequence of events then suggests that the following
 sequential simulation of $\v Z$ will entail the occurrence of the  (possibly rare) event\footnote{We denote the standard normal pdf with covariance $\Sigma$  via $\phi_\Sigma(\cdot)$ ($\phi(\cdot)\equiv \phi_{\m I}(\cdot)$) and the univariate cdf and complementary cdf by
$\Phi(\cdot)$ and $\overline\Phi(\cdot)$, respectively.}:
\begin{align*}
Z_1&\sim \frac{\phi(z_1)\bb I\{z_1\leq \alpha_1\}}{\Phi(\alpha_1)}\equiv \mathsf{TN}_{(-\infty,\alpha_1)}(0,1)\\
Z_2|Z_1&\sim \frac{\phi(z_2)\bb I\{z_2\leq \alpha_2(z_1)\}}{\Phi(\alpha_2(z_1))}\equiv \mathsf{TN}_{(-\infty,\alpha_2)}(0,1)\\
&\vdots\\
\hspace{-1cm}Z_d|Z_1,\ldots,Z_{d-1}&\sim \frac{\phi(z_d)\bb I\{z_d\leq \alpha_d(z_1,\ldots,z_{d-1})\}}{\Phi(\alpha_d(z_1,\ldots,z_{d-1}))}\equiv \mathsf{TN}_{(-\infty,\alpha_d)}(0,1)
\end{align*}
Denote the measure used to simulate $\v Z$  as $\bb P_{\v 0}$ and  the corresponding expectation (variance) operators as $\bb E_{\v 0}$ ($\var_{\v 0}$). 
With the above sampling scheme, the unbiased importance sampling estimator of $\ell$ (based on a single realization) is:
\begin{equation}
\label{simple}
\hat\ell_{\v 0}=\prod_{j=1}^d \Phi(\alpha_j(Z_1,\ldots,Z_{j-1})),\qquad \v Z\sim \bb P_{\v 0}
\end{equation}
Under the condition  that $\Sigma_{ii}<\Sigma_{ij}$ for some $i\not=j$ (see \cite{rached2017efficient}), the estimator \eqref{simple} is \emph{strongly efficient}, and thus 
 preferable to the Gulisashvili and Tankov  (GT) estimator \cite[Equation (65)]{gulisashvili2016tail}, which is only \emph{logarithmically efficient}. 
 The efficiency label stems from the fact that the relative error, $\var(\hat\ell_\mathrm{CMC})/\ell^2$, of the crude Monte Carlo estimator, 
\begin{equation*}
\label{CMC}
\hat\ell_\mathrm{CMC}=\bb I\{X_1+\cdots+X_d\leq \gamma\},\quad \ln\v X\sim\mathsf{N}(\v\nu,\Sigma),
\end{equation*}
 grows exponentially (in $\gamma$), while the relative  of the  GT estimator grows polynomially, and the relative  error of  \eqref{simple} decays to zero   as $\gamma\downarrow 0$, under the condition of Theorem~\ref{theorem:VRE left}.
\begin{thm}[Vanishing Relative Error]
\label{theorem:VRE left}
Suppose there exists an index $i$ such that $\Sigma_{ii}<\Sigma_{ij}$ for
all $i\not =j$, and, without loss of generality, assume that $i=1$.
Then, the estimator \eqref{simple}  enjoys the vanishing relative error property:
\[
\lim_{\gamma\downarrow 0}\frac{\var_{\v 0}\hat\ell_{\v 0}(\gamma)}{\ell^2(\gamma)}= 0
\]

\end{thm}

Although estimator \eqref{simple} can enjoy the best possible efficiency behavior, it is not necessarily  efficient when $\Sigma$ does not satisfy the condition in Theorem~\ref{theorem:VRE left}.
 To achieve  asymptotic efficiency for any $\Sigma$, we instead suggest the following parametric change of measure for $\v Z$, where the parameter $\v\mu$ still remains to be determined:
\begin{align*}
Z_1&\sim  \mathsf{TN}_{(-\infty,\alpha_1)}(\mu_1,1)\\
Z_2|Z_1&\sim \mathsf{TN}_{(-\infty,\alpha_2)}(\mu_2,1)\\
&\vdots\\
Z_d|Z_1,\ldots,Z_{d-1}&\sim\mathsf{TN}_{(-\infty,\alpha_d)}(\mu_d,1)
\end{align*}
Denote the measure used to simulate $\v Z$  as $\bb P_{\v\mu}$ and  the corresponding expectation (variance) operators as $\bb E_{\v\mu}$ ($\var_{\v\mu}$). 
Let the logarithm of the Radon-Nikodym derivative, $\m d\bb P/\m d\bb P_{\v\mu}$, be denoted as
\[
\begin{split}
\psi(\v z;\v\mu)&\idef \frac{\|\v\mu\|^2}{2}-\v z^\top\v\mu +\sum_{j=1}^d \ln\overline\Phi(\mu_j-\alpha_j(\v z)),
\end{split}
\]
and let $\c W=\{\v w: \v w\geq \v 0, \v 1^\top \v w=1\}$ denote the set of discrete probability distributions on $\bb R^d$.
 Then, our proposed unbiased estimator is 
\begin{equation}
\label{condest}
\hat\ell=\exp(\psi(\v Z;\v\mu^*)),\qquad \v Z\sim \bb P_{\v\mu^*}
\end{equation}
 where $\v\mu^*$ is  the solution to the  program: 
\begin{equation}\label{eq:optprogram}
(\v w^*,\v\mu^*)=\argmin_{\v w\in \c W,\v\mu} \left\{\|\v\mu\|^2+\ln
\overline\Phi\left( \frac{\v w^\top(\v\nu-\m L\v\mu)-\ln\gamma-\v w^\top\ln \v w}{\sqrt{\v w^\top\Sigma\v w}}\right)\right\}
\end{equation}
Why is  \eqref{condest} a good estimator? In addition to its  superior numerical performance (see Section~\ref{sec:GT}) compared to  the  Gulisashvili and Tankov  (GT) estimator \cite[Equation (65)]{gulisashvili2016tail}, it is also a logarithmically efficient estimator as $\gamma\downarrow 0$. This is formally stated in the
following theorem, which is proven in the appendix. 
\begin{thm}[Logarithmic Efficiency of Estimator]
\label{theorem:left tail}
The estimator \eqref{condest} is logarithmically efficient, that is,
\[
\liminf_{\gamma\downarrow 0} \frac{\ln\bb E_{\v\mu^*}\hat\ell^2(\gamma)}{\ln\ell(\gamma)}=2.
\]
with relative error
\footnote{The notation $f(x)\simeq g(x)$ as $x\rightarrow a$
stands for $\lim_{x\rightarrow a}f(x)/g(x)=1$. Similarly, we define $f(x)=\c O(g(x))\Leftrightarrow \lim_{x\rightarrow a}|f(x)/g(x)|<\mathrm{const.}<\infty$; $f(x)=o(g(x))\Leftrightarrow\lim_{x\rightarrow a}f(x)/g(x)=0$; also,
$f(x)=\Theta(g(x))\Leftrightarrow f(x)=\c O(g(x))\textrm{ and } g(x)=\c O(f(x))$. 
} that grows as $
\frac{\bb E_{\v\mu^*}\hat\ell^2(\gamma)}{\ell^2(\gamma)}= \c O( (-\ln\gamma)^{(d+1)}).
$
\end{thm}



A significant advantage of \eqref{condest} is that it is 
amenable to a randomized quasi Monte Carlo implementation \cite[Chapter 2, Algorithm 2.3]{kroese2011handbook}. This is because \eqref{condest} is 
a smooth infinitely differentiable estimator and as a result has finite \emph{Koksma-Hlawka} discrepancy bound
\cite[Chapter 2, Equation 2.3]{kroese2011handbook}. 
The advantage of smoothness even carries over to the estimator of the density of the SLN distribution.
The result is that we achieve significant variance reduction --- a point illustrated in the following sections.

\subsubsection{Numerical Comparison}
\label{sec:GT}
In this section we compare the performance  \eqref{condest} against the GT estimator \cite[Equation (65)]{gulisashvili2016tail}. In comparing relative performance, we use the (estimated) relative error in percentage,
$
\mathrm{RE}(\widehat{\ell})=\sqrt{\var(\hat{\ell} )/n}/\ell
$
and {\em work-normalized relative variance},
$ 
\mathrm{WNRV}(\hat{\ell})=\mathrm{RE}^2(\widehat{\ell})\times (\textrm{total computing time in seconds})
$.

Table~\ref{tab:cdf} shows the numerical results using a sample size of
$n=10^6$ and the parameters
 $d=20,\v \nu = \v 0,\Sigma = {\rm diag}(\v \sigma)$, where $\sigma_k^2 = k$. 
\begin{table}[H]
\centering
\caption{Results for $d=20,\v \nu = \v 0,\Sigma = {\rm diag}(\v \sigma)$, where $\sigma_k^2 = k$.} 
\begin{tabular}{c|c|c|c|c|c|c}

 $\gamma$ & $\widehat{\ell}$ & $\widehat{\ell}_{\rm GT}$&${\rm RE}(\widehat{\ell})\%$&${\rm RE}(\widehat{\ell}_{\rm GT})\%$& {WNRV}$(\widehat{\ell})$ & WNRV$(\widehat{\ell}_{\rm GT})$ \\ \hline 
12& \num{1.68E-04} & \num{1.67E-04} & 0.198 &4.81&\num{2.37E-05} & \num{1.75e-03} \\
10& \num{6.82E-05} & \num{6.88E-05} & 0.217 &6.66&\num{2.84E-05} & \num{3.26e-03}  \\
8& \num{2.01E-05} & \num{2.02E-05} & 0.244 &4.91&\num{3.69E-05} & \num{1.85e-03} \\
6& \num{3.54E-06} & \num{3.46E-06} & 0.285 &5.17&\num{5.00E-05} & \num{2.00e-03}  \\ 
4& \num{2.13E-07} & \num{2.17E-07} & 0.368 &5.46&\num{8.44E-05} & \num{2.27e-03}  \\
3& \num{2.20E-08}& \num{2.35E-08} & 0.439 &6.89&\num{1.21E-04} & \num{3.65e-03}  \\  
2& \num{6.05E-10} & \num{5.63E-10}& 0.567 &10.9&\num{2.04E-04} & \num{9.19e-03}   \\
1& \num{4.24E-13} & \num{4.31E-13}& 0.937 &17.8&\num{5.47E-04} & \num{2.35e-02}   
  \end{tabular}
	\label{tab:cdf}
\end{table}
The results are self-explanatory --- we can see the that for $\gamma=1$, the relative error of the GT estimator is  large.

  In our numerical simulations  we observe that the GT
estimator performs at its best 
 when all $\nu$'s   are the same, and otherwise it
may not perform so well. For example, in Table~\ref{tab:cdftwo} the relative error 
is  larger, because we use the different means $\nu_k = k - d,\;k=1,\ldots,d$.

\begin{table}[H]
\centering
\caption{Results for $\Sigma = {\rm diag}(\v \sigma),\nu_k = k - d,\sigma_k^2 = k,d=10$.} 
\begin{tabular}{c|c|c|c|c|c|c}
\multicolumn{3}{c}{} & \multicolumn{2}{|c}{relative error \%}& \multicolumn{2}{|c}{work normalized rel. var.}  \\
 \hline
$\gamma$ & $\widehat{\ell}$ &   $\widehat{\ell}_{\rm GT}$&${\rm RE}(\widehat{\ell})$&${\rm RE}(\widehat{\ell}_{\rm GT})$& {WNRV}$(\widehat{\ell})$ & WNRV$(\widehat{\ell}_{\rm GT})$ \\ \hline 
\num{1} & \num{1.25E-01} & \num{5.47E-09} & \num{0.0389} & \num{41} & \num{4.68E-07} & \num{6.58e-02}   \\
\num{1E-01} & \num{2.75E-03} & \num{5.39E-05} & \num{0.0956} & \num{51.4} & \num{2.82E-06} & \num{1.02e-01} \\ 
\num{1E-02} & \num{7.10E-07} & \num{7.47E-07} & \num{0.209} & \num{38} & \num{1.33E-05} & \num{5.67e-02} \\
\num{1E-03} & \num{8.59E-14} & \num{8.13E-14} & \num{0.466} & \num{7.58} & \num{6.80E-05} & \num{2.29e-03}  \\
\num{1E-04} & \num{1.03E-25} & \num{1.07E-25} & \num{0.967} & \num{9.68} & \num{2.99E-04} & \num{3.77e-03}   \\
\num{1E-05} & \num{1.10E-43} & \num{8.92E-44} & \num{1.79} & \num{11.9} & \num{1.01E-03} & \num{5.49e-03} \\
\num{1E-06} & \num{4.27E-68} & \num{2.61E-68} & \num{2.81} & \num{14.2} & \num{2.48E-03} & \num{8.03e-03}   
  \end{tabular}
\label{tab:cdftwo}
\end{table}

In the above setting, it appears that the accuracy of the GT estimator initially deteriorates before it improves.
 One explanation for this phenomenon is that the asymptotic approximation upon which the GT estimator is built is poor in a non-asymptotic regime -- a point revisited in Section~\ref{sec:asymp}.

We observe that both estimators benefit
 from positive correlation. For example, if we take $\v \nu$ to be a linearly spaced vector on the interval $[0,1/4]$ with $d=50$, and $\rho =0.25,\Sigma=0.25^2(\rho \v 1 \v 1^\top + (1-\rho)\m I)$, then Table~\ref{tab:cdf3}
shows the slowly increasing relative error for both estimators as $\gamma$ becomes smaller. Again, observe that the variance of the new estimator \eqref{condest}
is typically  more than a hundred times smaller. 
\begin{table}[H]
\centering
\caption{Results for covariance matrix with positive correlation.} 
\begin{tabular}{c|c|c|c|c|c|c}
\multicolumn{3}{c}{} & \multicolumn{2}{|c}{relative error \%}& \multicolumn{2}{|c}{work normalized rel. var.}  \\
 \hline
 $\gamma$ & $\widehat{\ell}$ &   $\widehat{\ell}_{\rm GT}$&${\rm RE}(\widehat{\ell})$&${\rm RE}(\widehat{\ell}_{\rm GT})$& {WNRV}$(\widehat{\ell})$ & WNRV$(\widehat{\ell}_{\rm GT})$ \\ \hline 
 40 & \num{1.85E-03} & \num{1.86E-03} & \num{0.169} & \num{2.45} & \num{4.41E-05} & \num{1.07e-03}    \\
38 & \num{4.83E-04} & \num{5.10E-04} & \num{0.178} & \num{3.24} & \num{4.92E-05} & \num{1.93e-03}   \\
36 & \num{9.96E-05} & \num{9.63E-05} & \num{0.189} & \num{2.01} & \num{5.49E-05} & \num{7.57e-04} \\
34 & \num{1.57E-05 }& \num{1.56E-05} & \num{0.199} & \num{4.47} & \num{6.09E-05} & \num{3.62e-03}   \\
32 & \num{1.79E-06} & \num{1.89E-06} & \num{0.209} & \num{7.62} & \num{6.78E-05} & \num{1.06e-02} \\
30 & \num{1.41E-07} & \num{1.36E-07} & \num{0.219} & \num{2.86} & \num{7.43E-05} & \num{1.52e-03}   \\ 
28 & \num{7.06E-09} & \num{7.10E-09} & \num{0.23} & \num{2.70} & \num{8.25E-05} & \num{1.32e-03} \\
26 & \num{2.09E-10} & \num{2.13E-10} & \num{0.241} & \num{4.11} & \num{9.00E-05} & \num{3.18e-03}  \\
24 & \num{3.25E-12} & \num{3.37E-12} & \num{0.251} & \num{3.13} & \num{9.73E-05} & \num{1.84e-03}   \\
22 & \num{2.28E-14} & \num{2.42E-14} & \num{0.263} & \num{3.83} & \num{1.08E-04} & \num{2.66E-03}
 \end{tabular} 
\label{tab:cdf3}
\end{table}

Finally, we verify that when $\Sigma$ satisfies the property in Theorem~\ref{theorem:VRE left}, we obtain vanishing relative error. Consider the numerical
example from \cite{rached2017efficient}, where $\nu=(4,4,4,4)^\top$ and 
\[
\Sigma=\left[\begin{array}{cccc} 
1&2&2&2\\
2&5&4&4\\
2&4&4.5&4\\
2&4&4&4.5
\end{array}\right]
\]
This $\Sigma$ satisfies the property that $\Sigma_{11}<\Sigma_{1j}$ for all $j\not =1$.
Table~\ref{VRE table} shows that in this case the gains from using the strongly efficient estimator are significant --- the relative error is easily more than a thousand times smaller.
\begin{table}[H]
\centering
\caption{Comparison between the strongly efficient estimator \eqref{simple} and the weakly efficient estimator\eqref{condest}.}  \vspace{5mm}
\begin{tabular}{c|c|c|c|c}
\multicolumn{3}{c}{} & \multicolumn{2}{|c}{relative error \%} \\
\hline
 $\gamma$ &  $\widehat{\ell}_{\v 0 }(\gamma)$ &  $\widehat{\ell}(\gamma)$ & ${\rm RE}(\widehat{\ell}_{\v 0})$ & ${\rm RE}(\widehat{\ell})$ \\ \hline
 \num{10} & \num{1.91e-02} & \num{1.91e-02} & \num{1.04e-03} & \num{9.82e-04} \\
 \num{1} & \num{2.40e-05} & \num{2.39e-5} &  \num{5.04e-04} & \num{1.37e-03} \\
 \num{1e-1} & \num{1.39e-10} & \num{1.39e-10} & \num{1.99e-04} & \num{1.82e-3} \\ 
\num{1e-2} & \num{3.78e-18} & \num{3.79e-18 }& \num{7.29e-05} & \num{2.18e-3}\\ 
\num{1e-3} & \num{5.29e-28} & \num{5.29e-28} & \num{2.65e-05} & \num{2.49e-3}\\ 
\num{1e-4} & \num{3.82e-40} & \num{3.83e-40} & \num{8.78e-06} & \num{2.76e-3}\\ 
\num{1e-5} & \num{1.42e-54} & \num{1.42e-54} & \num{3.35e-06} & \num{3.01e-3}\\
\num{1e-6} & \num{2.68e-71} & \num{2.68e-71} & \num{1.01e-06} & \num{3.23e-3}\\ 
  \end{tabular}
	\label{VRE table}
\end{table}

\subsubsection{Acceleration via Quasi Monte Carlo}
As mentioned previously, a significant advantage of  estimator \eqref{condest} is that 
it is a smooth infinitely differentiable function, amenable 
to acceleration using quasirandom sequences \cite[Chapter 2]{kroese2011handbook}. While the
 standard error of a Monte Carlo estimator, driven by pseudorandom numbers,  decays at the canonical rate of $\c O(n^{-1/2})$, 
 the standard error of a Quasi Monte Carlo estimator, driven by quasirandom numbers,  decays at the superior rate of $\c O(n^{-1/2-\delta})$ for some $\delta>0$
that depends on the dimension $d$ and the  smoothness of the estimator. 

 Figure~\ref{fig:cdf} below shows that for $d=5$, the rate of decay of \eqref{condest}  improves from the canonical rate of  $\c O(n^{-0.5})$ (when using a pseudorandom
sequence) to approximately $\c O(n^{-0.93})$ when using  Sobol's quasirandom sequence
\cite[Section 2.5]{kroese2011handbook}. Here the relative error is estimated using 
$100$ independent random shifts of Sobol's quasirandom pointset \cite[Section 2.7]{kroese2011handbook}, and the number $-0.93$ is simply the slope of the line of best fit.
\begin{figure}[H]
	 \caption{The relative error of estimator \eqref{condest} for $\gamma=\bb E S=5 \exp(1/2), d=5$, 
	$\Sigma=\m I,\v\nu=\v 0$. Also displayed is a reference line with the canonical slope of $-1/2$.} 
 \begin{center}
	\includegraphics[scale=0.6]{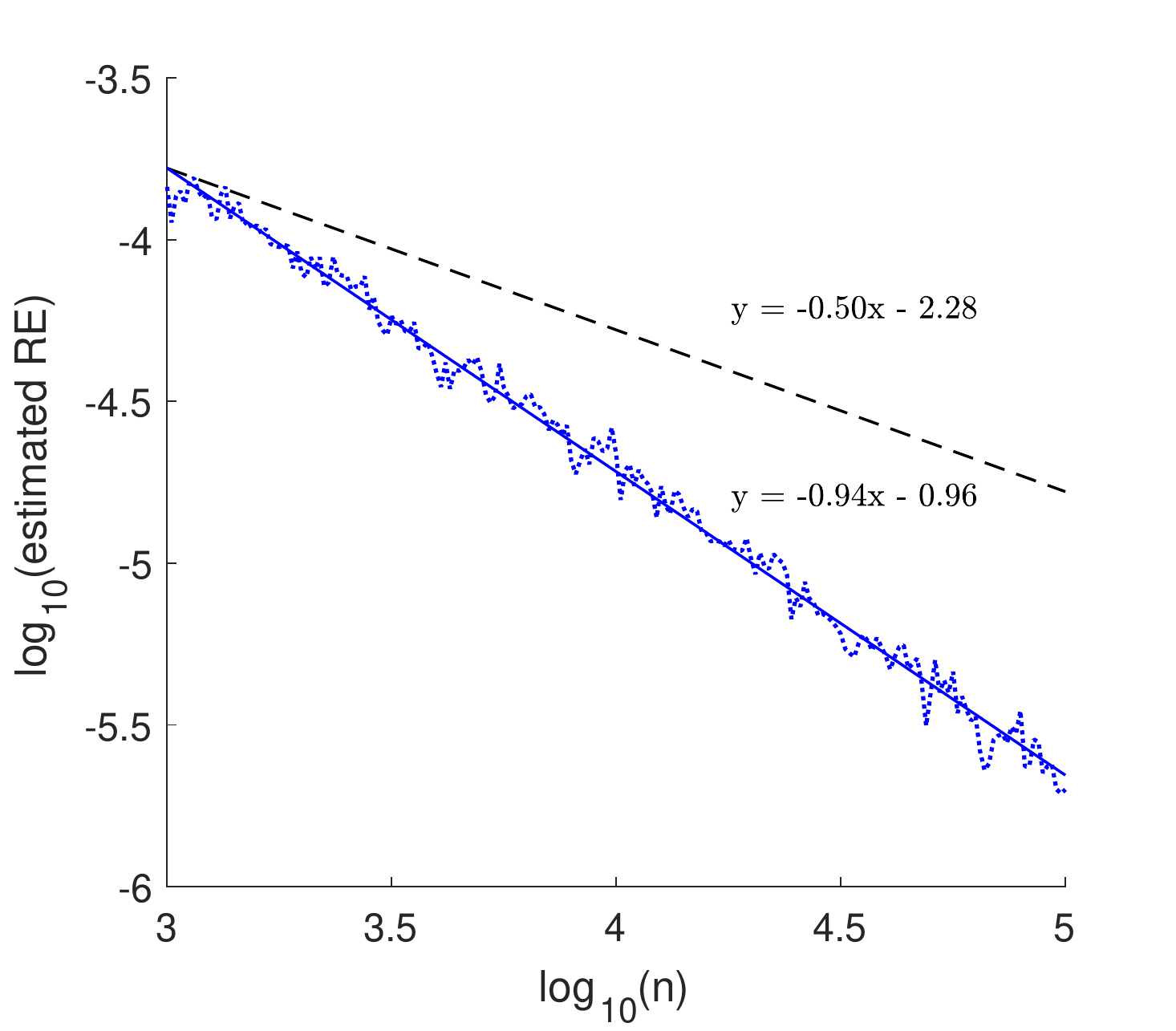}
\end{center}
		\label{fig:cdf}
\end{figure}

\subsection{Density Estimator} To derive the smooth density estimator, we use the so-called
 \emph{push-out method} \cite[Chapter 11]{kroese2011handbook}. In particular, observe that we can ``push-out" $\gamma$ as follows:
\[\textstyle
\ell(\gamma)=\bb P(\exp(\nu_1-\ln(\gamma)+L_{11}Z_1)+\cdots+
\exp(\nu_d-\ln(\gamma)+\sum_{j=1}^dL_{dj}Z_j)\leq 1)
\]
 Therefore,  the pdf of the SLN distribution can written as the integral:
\[
\begin{split}
f(\gamma)=\frac{\partial\ell}{\partial\gamma}&=\int_{\sum_i\exp(u_i)<1}\frac{\partial}{\partial\gamma}\phi_\Sigma(\v u-\v\nu+\v 1\ln\gamma)\m d\v u\\
&=\int_{\sum_i\exp(u_i)<1}\phi_\Sigma(\v u-\v\nu+\v 1\ln\gamma)\frac{-\v 1^\top\Sigma^{-1}(\v u-\v\nu+\v 1\ln\gamma)}{\gamma}\m d\v u\\
&=
\int_{\sum_i\exp(u_i)<\gamma}\phi_\Sigma(\v u-\v\nu)\frac{-\v 1^\top\Sigma^{-1}(\v u-\v\nu)}{\gamma}\m d\v u\\
&=\int_{\bb R^d}\phi(\v z)\frac{-\v z^\top\m L^{-1}\v 1}{\gamma}\textstyle\bb I\{\exp(\nu_1+L_{11}z_1)+\cdots+\exp(\nu_d+\sum_{j=1}^dL_{dj}z_j)<\gamma\}\m d\v z
\end{split}
\]
As a result of this, our smooth unbiased estimator of the SLN pdf is:
\begin{equation}
\label{pdf est}
\hat f(\gamma)=\exp(\psi(\v Z;\v\mu^*))\frac{-\v Z^\top\m L^{-1}\v 1}{\gamma},\qquad \v Z\sim \bb P_{\v\mu^*}
\end{equation}
\subsubsection{Numerical Examples}
Using  estimator \eqref{pdf est}, we can  estimate accurately the effect of the correlation coefficient
 $\rho$ on the shape of the SLN pdf. The figure below shows
that as $\rho$ increases the tail of the SLN pdf becomes thicker.  
\begin{figure}[H]
\caption{Estimate of the SLN pdf for $d=32,\v\nu=\v 0$, $\Sigma =\rho\v 1\v 1^\top +(1-\rho)\m I$, and varying $\rho$.}
\centering
\includegraphics[width=\linewidth]{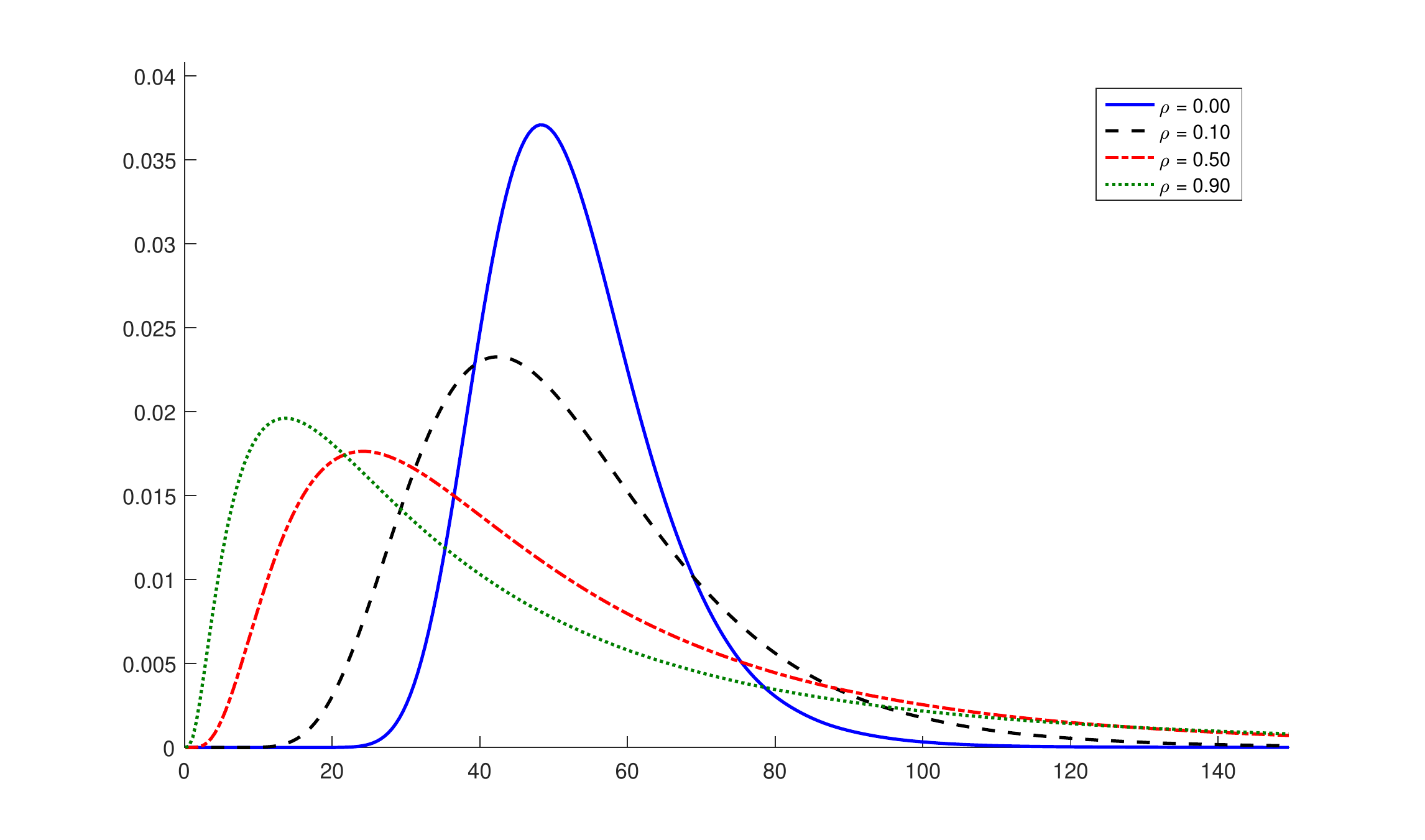}
\end{figure}

We now comment on what appears to be the only other viable Monte Carlo estimator of the SLN density.

The pdf estimator proposed in 
\cite[Equation 13]{asmussen2016exponential}  works only when the log-normal factors are independent, but one can extend it to the dependent case as shown in  \cite{asmussen2017conditional}. Let us denote the  estimator proposed in \cite{asmussen2017conditional,asmussen2016exponential} by $\hat\ell_{\rm A}$.
  Tables~\ref{tab:pdf}  and \ref{tab:pdf2} below compare the two estimators on two distinct numerical examples. The results suggest that   \eqref{pdf est} becomes significantly more efficient than $\hat\ell_{\rm A}$ when the $X_i$'s have different marginal distributions. Note that, as expected, the efficiency of \eqref{pdf est} deteriorates as we approach the right tail (the right tail requires a different approach, see Section~\ref{sec:right-tail}). 
\begin{table}[H]
\centering
\caption{The SLN distribution for $d=32,\v\nu=\v 0,\rho=0.5$, $\Sigma =\rho\v 1\v 1^\top +(1-\rho)\m I$, using $n=10^6$ samples.} \vspace{5mm}
\begin{tabular}{c|c|c|c|c|c|cs}
\multicolumn{3}{c}{} & \multicolumn{2}{|c}{relative error \%}& \multicolumn{2}{|c}{work normalized rel. var.}  \\
 \hline
 $\gamma$ & $\hat\ell(\gamma)$ &  $\widehat{f}(\gamma)$  &  RE($\widehat{f}$) &  RE($\widehat{f}_{{\rm A}}$) &  WNRV($\widehat{f}$) &  WNRV($\widehat{f}_{{\rm A}}$) \\ \hline

 140 & \num{0.957} & \num{9.12e-04} & 0.960 & 0.914  & \num{7.93e-04} &  \num{1.50e-03}    \\
 100 &  \num{0.894} & \num{2.53e-03} & 0.421 & 0.643   &  \num{1.52e-04} &  \num{7.42e-04}  \\
80 &  \num{0.826} & \num{4.46e-03} & 0.260 & 0.538    & \num{5.84e-05} & \num{5.05e-04}   \\
 60 & \num{0.705} & \num{7.96e-03} & 0.151 & 0.462   &  \num{2.00e-05} & \num{3.61e-04}  \\
50 &  \num{0.613} & \num{1.06e-02} & 0.113 & 0.436   & \num{1.14e-05}& \num{3.23e-04}  \\
40 & \num{0.490} & \num{1.38e-02} & 0.090 & 0.426 & \num{7.05e-06} & \num{3.08e-04} \\
30 &  \num{0.336} & \num{1.69e-02} & 0.084 & 0.444   & \num{6.23e-06} & \num{3.31e-04}   \\
20 & \num{0.163} & \num{1.71e-02} & 0.098 & 0.543  & \num{8.70e-06} & \num{4.94e-04} \\
15 &  \num{0.831} & \num{1.41e-02} & 0.113 & 0.693  & \num{1.17e-05} &  \num{7.99e-04} \\
 \end{tabular}
\label{tab:pdf}
\end{table}

\begin{table}[H]
\centering
\caption{The SLN distribution for $d=10, \rho=0$, $\nu_i=i-d$, $\sigma_i^2 =i$, estimated with $n=10^6$ samples.} \vspace{5mm}
\begin{tabular}{c|c|c|c|c|c|c}
\multicolumn{3}{c}{} & \multicolumn{2}{|c}{relative error \%}& \multicolumn{2}{|c}{work normalized rel. var.}  \\
 \hline
 $\gamma$ & $\hat\ell(\gamma)$ &  $\widehat{f}(\gamma)$  &  RE($\widehat{f}$) &  RE($\widehat{f}_{{\rm A}}$) &  WNRV($\widehat{f}$) &  WNRV($\widehat{f}_{{\rm A}}$) \\ \hline 

500 & \num{0.964} & \num{5.28e-05} & 6.22 & 12.0 & \num{1.09e-02} & \num{2.60e-02}    \\
 100 & \num{0.881} & \num{8.01e-04} & 1.86 & 30.2  & \num{9.91e-04} & \num{1.66e-01}  \\
 30 & \num{0.746} & \num{4.81e-03} & 0.88 & 13.3  & \num{2.08e-04} &  \num{3.23e-02}  \\
 15 & \num{0.633} & \num{1.21e-02} & 0.59 & 7.23  &  \num{9.56e-05} &  \num{9.52e-03}   \\
 7 & \num{0.484} &  \num{2.96e-02} & 0.39 & 5.26 &  \num{4.26e-05} &  \num{5.05e-03}   \\
 3 & \num{0.310} & \num{6.58e-02} & 0.27 & 3.51  & \num{2.09e-05} &    \num{2.22e-03}  \\
 1 & \num{0.125} & \num{1.29e-01} & 0.17 & 2.42  & \num{8.86e-06} & \num{1.03e-03}    \\
 0.5 & \num{0.0548} & \num{1.50e-01} & 0.14 & 2.24  & \num{5.56e-06} & \num{9.12e-04}  \\
 \end{tabular}
\label{tab:pdf2}
\end{table} 

Finally, 
 we confirmed that, as expected, quasi Monte Carlo again accelerates the speed of convergence of the smooth estimator \eqref{pdf est} by as much as approximately $\c O(n^{-0.92})$. 
The qualitative behavior 
is depicted on Figure~\ref{fig:pdf}, where we also show the rate of convergence of
 its competitor $\hat\ell_{\rm A}$. Here, again, the relative error (in percent) is  estimated using  $100$ independent random shifts of Sobol's quasirandom pointset \cite[Section 2.5]{kroese2011handbook}.

The reason that estimator \eqref{pdf est} achieves a better convergence rate  
is  that, while  $\hat\ell_{\rm A}$ is  continuous, but not differentiable, the estimator \eqref{pdf est} is infinitely differentiable, and  hence more amenable to acceleration with quasirandom sequences.   
\begin{figure}[H]
	 \caption{The relative error of estimator \eqref{condest} (in blue with slope -0.92) for $\gamma=\bb E S=5 \exp(1/2), d=5$, 
	$\Sigma=\m I,\v\nu=\v 0$, as well as that of $\hat\ell_{\rm A}$ (in red with slope -0.67).}
 \begin{center}
	\includegraphics[scale=0.7]{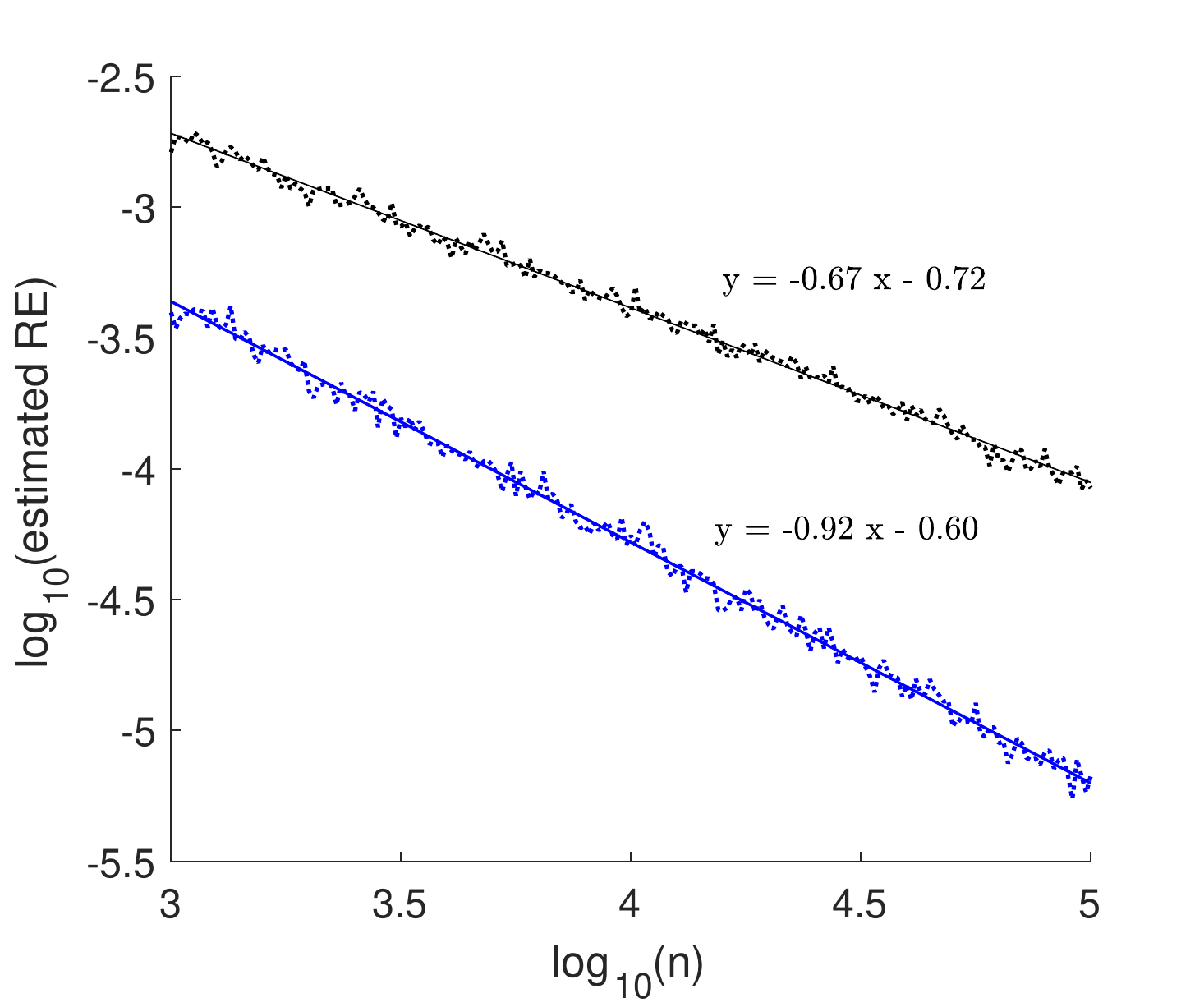}
\end{center}
		\label{fig:pdf}
\end{figure}



%
%

\subsection{Perfect Simulation from Conditional Distribution}
\label{sec:alg}
One of advantages of our approach is that when $d$ is not too large, for the first time,  it is possible to simulate exactly  from the distribution of $\v X$, conditional on the rare-event $\left\{\sum_{k=1}^d X_k\le \gamma\right\}$. 
As $\psi(\v z; \v \mu^*)$ is concave in $\v z$ for any fixed $\v \mu^*$ (see, for example, \cite[Lemma 1]{botev2017normal}), we can easily obtain its  maximum. This can then be used in the following acceptance-rejection sampling procedure.

\begin{enumerate}
\item {\bf Require:} $c = \max_{\v z}\psi(\v z; \v \mu^*)$ 
\item {\bf Until}  $E > c-\psi(\v Z; \v \mu^*)$ {\bf do} \\
 Simulate $\v Z \sim \bb P_{\v \mu^*}$ and $E \sim {\sf Exp}(1)$, independently.

\item {\bf Return} $\v X = \exp(\v \nu + \m L \v Z)$ as a sample from the conditional distribution.
\end{enumerate}
We next exploit the ability to simulate from the conditional SLN distribution 
to simulate exactly the  stock price trajectories  of an Asian option
with positive payoff. 
\subsubsection{Exact Simulation of Stock Prices under Black-Scholes}
 We consider the average value of a stock price observed at a set of discrete times on the interval $[0,T]$, 
\[ 
\bar{X}_T=\frac{1}{d+1}\sum_{i=0}^{d}X_{t_i}, \quad t_0=0 < t_1 < t_2 < \cdots < t_d = T,
\]
under the Black-Scholes model 
$X_t = X_0 \exp\left((r-\sigma^2/2)t + \sigma W_t\right)$, where:
(1) $W_t$ is the Wiener process at time $t$; (2)  $\sigma$ is the volatility coefficient; (3)  $r$ is the risk-free interest rate. Then, for an Asian Put option  with maturity $T$ and strike price $K$ the payout is  $(K - \bar{X}_T)^+$. Since $\v X=\left(X_{t_1},X_{t_2}, \ldots, X_{t_d}\right)$  is a  log-normal random vector with 
\[ \begin{split}
& \nu_i = \ln(X_0)+(r-\sigma^2/2)t_{i}, \quad i=1,\ldots,d  \\ 
& \Sigma_{ij} = \sigma^2\min\left(t_i, t_j \right), \quad i,j=1,\ldots,d,
\end{split}
\]
 we can use our algorithm above 
to simulate a realization of the stock price path conditional on the event $\{\bar{X}_T<K\}=\{X_{t_1}+\cdots+X_{t_d}<(d+1)K - X_{0}\}$. Simulation of such an $\v X$ conditional on  the rare-event $\{\bar{X}_T \le 30\}$ provides insight  into how the rare event occurs, that is,  how the stock price must behave for a positive payoff.

The following figure shows one stock price realization with parameters
$X_0 = 50, K=30, \sigma = 0.25,r=0.07, T=4/12, d=88$.
\begin{figure}[H]
\begin{center}
\caption{Stock price trajectory, conditional on  $\bar{X}_T \le 30$. }   
	\includegraphics[width=\linewidth]{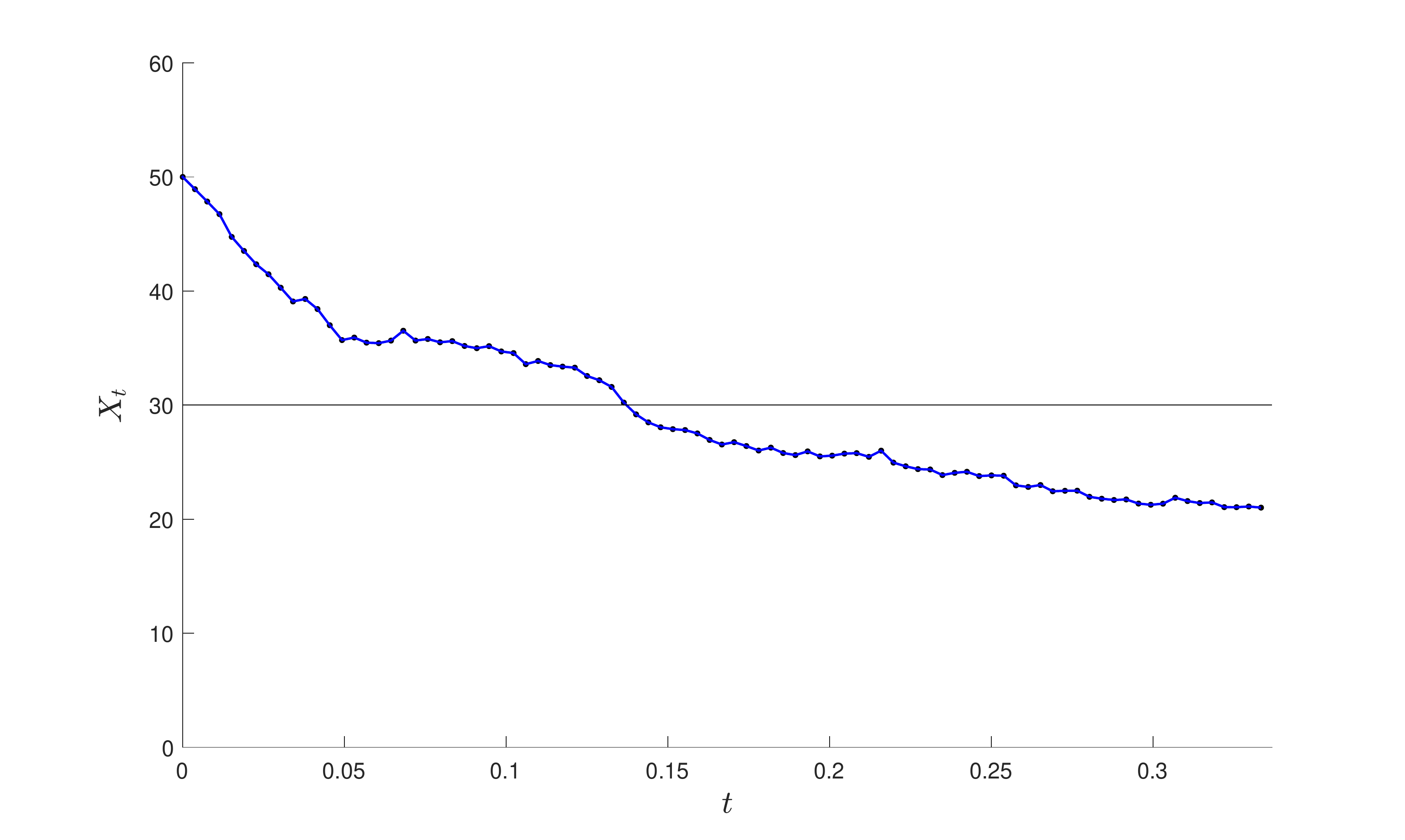}
	\end{center}
	\end{figure}
We note that, since  $\bb P(\bar{X}_T \le 30)\approx 2\times 10^{-11}$ is a rare-event probability, exact simulation of such a stock price trajectory is not possible using a naive acceptance-rejection, because the acceptance rate would be approximately $2\times 10^{-11}$. Instead, our  algorithm enjoys the (estimated) acceptance rate of $5.9\%$.


\section{Accurate Estimation of the Right Tail}
\label{sec:right-tail}
In this section, we provide an  estimator of   the right tail of the SLN distribution, that is,
\[
\bb P(X_1+\cdots+X_d\geq \gamma),\qquad \ln\v X\sim\mathsf{N}(\v\nu,\Sigma)
\] 
that works well for many parameter settings for which all existing estimators  fail. 
In  order to keep the notation minimal,  we recycle the notation for the left tail and henceforth let
\[
\ell(\gamma)\idef\bb P(\exp(Y_1)+\cdots\exp(Y_d)\geq\gamma),\qquad \v Y\sim\mathsf{N}(\v\nu,\Sigma)
\]
Further, we set
$S=X_1+\cdots+X_d, M=\max_i X_i$, and
$\sigma_i^2=\Sigma_{ii}, \sigma=\max_k \sigma_k, \nu=\max\{\nu_k: \sigma_k=\sigma\}$. Note that with all random variates defined on the same probability space,  we can write $\v X=\exp(\v Y)$.

One of the reasons why estimating the right tail is difficult is due to the
heavy-tailed behavior  of $\ell(\gamma)$ as $\gamma\uparrow\infty$ (see  Corollary~\ref{cor:asymp} here or \cite{asmussen2008asymptotics}):
\begin{equation}
\label{asymp}
\ell(\gamma)\simeq \ell_\mathrm{as}\idef\sum_{k=1}^d\bb P(Y_i\geq\ln\gamma)=\sum_{k=1}^d \overline \Phi((\ln\gamma-\nu_k)/\sigma_k)
\end{equation}
To tackle this problem, the authors of  \cite{asmussen2011efficient,blanchet2008efficient,gulisashvili2016tail,kortschak2013efficient} propose a number of theoretically efficient estimators. The problem with these estimators, however, is that their established theoretical efficiency  does not
 necessarily translate into  estimators with reasonably low Monte Carlo variance. Before proceeding to remedy this problem, we next explain why these existing proposals can fail to estimate $\ell(\gamma)$ --- we provide both numerical evidence and theoretical insight.


\subsection{Variance Boosted Estimator} We call the first estimator proposed in the literature   the \emph{variance boosted} estimator \cite{asmussen2011efficient,blanchet2008efficient}. It is defined as follows.

Let $\bb P_\theta$ be an importance sampling measure under which 
$\v Y\sim \mathsf{N}(\v\nu,\Sigma/(1-\theta))$ for some parameter $\theta\in[0,1)$.
If we take $\theta$ sufficiently close to unity, then we can inflate the variance of $\v Y$  to induce the event $\{S>\gamma\}$. We thus obtain the variance boosted estimator:
\begin{equation}
\label{var boost}
\hat\ell_\theta(\gamma)=\frac{\exp(-\theta (\v Y-\v \nu)^\top\Sigma^{-1}(\v Y-\v\nu)/2)}{(1-\theta)^{d/2}}\bb I\{S>\gamma\},\qquad \v Y\sim \bb P_\theta
\end{equation}
One can choose $\theta$ optimally and show \cite{asmussen2011efficient} that:
\[
\frac{\bb E_\theta \hat\ell^2_\theta}{\ell^2}=\Theta([\ln\gamma]^{d/2+1}\gamma^{1/4})
\]
Therefore, we expect that the variance boosted estimator will only be useful for very small $d$ and small $\gamma$. In contrast, in Section~\ref{sec:new}
we show that our new proposal has relative error which grows at the much slower rate of $\Theta(\ln\gamma)$.

Consider a simple example in which all log-normals are iid with 
 $\sigma=0.25,\Sigma=\m I \times \sigma^2,\v\nu=\v 0,d=30$. 
Table~\ref{tab:var boost} shows the estimated values for $\ell(\gamma)$
for different values of $\gamma$ using three different estimators:
the variance boosted  $\hat\ell_\theta$; the \emph{Asmussen-Kroese estimator}  \cite{asmussen2006improved},
\begin{equation}
\label{AKest}
\textstyle
\hat\ell_\mathrm{AK}=d\overline\Phi\left(\frac{1}{\sigma}\ln\left[(\gamma-\sum_{j<d}X_j) \vee \max_{j<d}X_j \right]\right),\quad \ln\v X\sim\mathsf{N}(\v 0, \sigma^2\m I) ;
\end{equation}
 and our proposed estimator $\hat\ell$ in Section~\ref{sec:new}. The data was populated using $n=10^7$ independent replications of each estimator. The difference in the CPU run times for all methods was negligible (all between 7 to 10 seconds), and hence not reported here. The  conclusion from the results in the table is that the variance boosted estimator, $\hat\ell_\theta$, is not useful due to its high variability. 


\begin{table}[H]

\caption{Comparative performance of the variance-boosted  and  Asmussen-Kroese estimators. The proposed estimator $\hat\ell$ is given in column two and described in Section~\ref{sec:new}.}
\label{tab:var boost}
\begin{center}
\begin{tabular}{c|c|c|c|c|c}  
\multicolumn{3}{c}{} & \multicolumn{3}{|c}{relative error \%}  \\
 \hline
$\gamma$& $\hat\ell$  & $\hat\ell_\mathrm{AK}$ &$\mathrm{RE}(\hat\ell)$  & $\mathrm{RE}(\hat\ell_\mathrm{AK})$ & $\mathrm{RE}(\hat\ell_\theta)$  \\
\hline
30&	0.74	  &    0.74    &0.199 &	0.0321&	0.314 	\\	
33&	0.079   & 	   0.079    &0.26  &	0.0871&	3.67	\\
36&	0.00052 &	      0.00052 &0.403 &	0.684&	39.8 	\\
39&	\num{2.94e-07}&	\num{3.31e-07}  &0.725&	17.9&	51.9 		\\
42&	\num{2.29e-11}&	\num{9.23e-14}  &1.45&	54.6&	99.9 		\\
45&	\num{3.92e-16}&	\num{7.78e-20} &2.57&	64.4&	97.8	\\
48&	\num{1.93e-21}&	\num{2.13e-25}  &4.44&	31.7&	97	\\
51&	\num{3.98e-27}&	\num{2.40e-29}  &7.85&	25.2&	81.5 		\\
54&	\num{8.58e-33}&	\num{3.96e-33}  &3.22&	15.3&	100	\\
57&	\num{3.44e-36}&	\num{3.07e-36}  &0.418&	13.3&	69.8 	\\	
60&	\num{4.26e-39}&	\num{3.86e-39}  &0.203&	5.21&	99.7 	\\
63&	\num{1.06e-41}&	\num{1.01e-41}  &0.18&	2.92&	99		\\
66&	\num{4.38e-44}&	\num{4.39e-44}  &0.162&	1.58&	64.8 \\
69&	\num{2.75e-46}&	\num{2.74e-46}  &0.16&	1.09&	100		\\
72&	\num{2.42e-48}&	\num{2.40e-48}  &0.155&	0.686&	98.3  \\	
75&	\num{2.83e-50}&	\num{2.81e-50}  &0.153&	0.498&	72.1 	\\
78&	\num{4.24e-52}&	\num{4.21e-52}  &0.151&	0.414&	95.7 	\\
81&	\num{7.87e-54}&	\num{7.86e-54}  &0.15&	0.287	&99.3		\\
84&	\num{1.78e-55}&	\num{1.78e-55}  &0.15&	0.26	&100 	\\
87&	\num{4.74e-57}&	  \num{4.75e-57} &0.15&	0.251	&90.5  	\\
90&	\num{1.48e-58}&	  \num{1.48e-58} &0.15&	0.189	&100	
\end{tabular}
\end{center}
\end{table}

\begin{figure}[H]
\caption{The estimated relative error of $\hat\ell_\theta$ as a function of $\theta$ using $10^7$ replications. The smallest
estimated relative error was $23\%$, corresponding to $\theta=0.71$. Where the estimate of $\ell(45)$ is $0$, the  relative error is recorded as unity (100\%).}
\begin{center}
\includegraphics[scale=0.6]{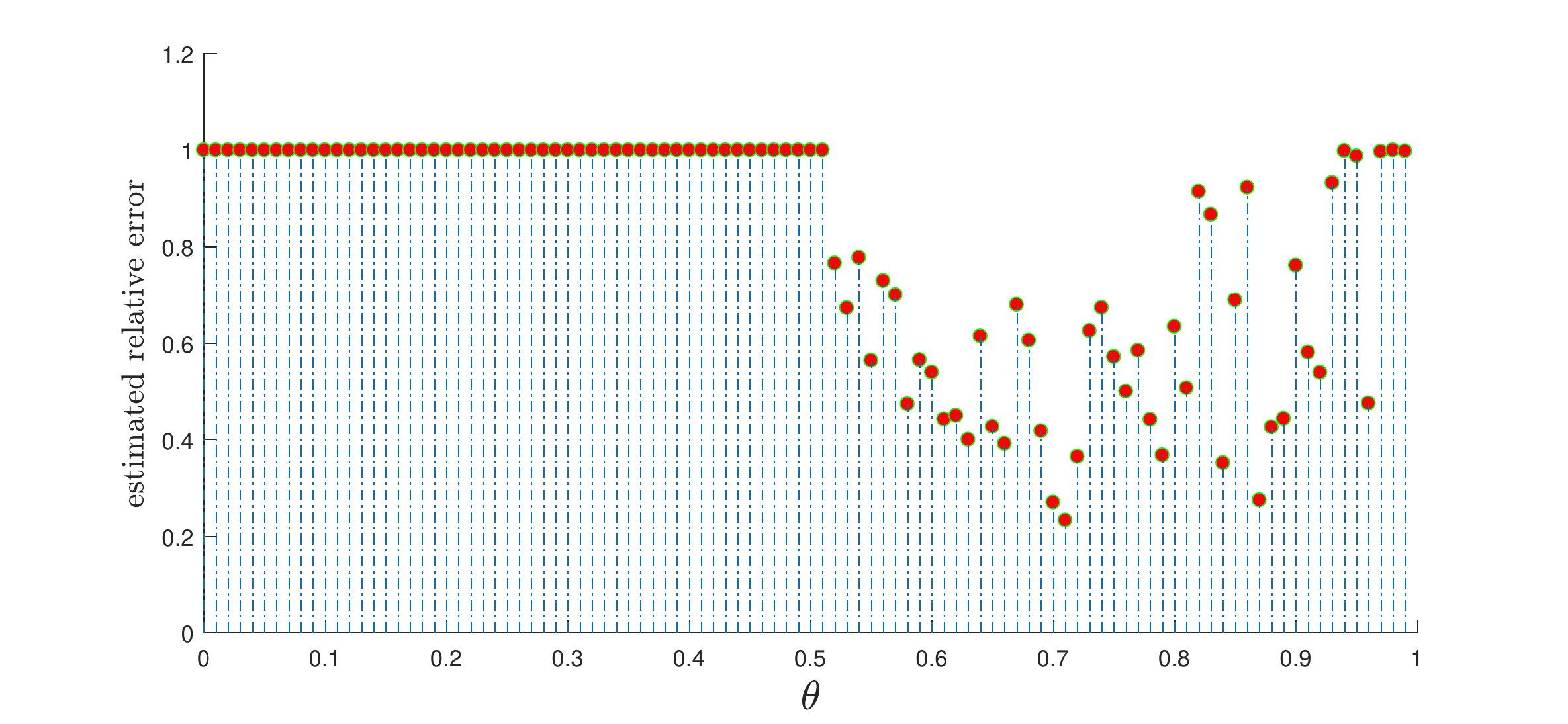}
\end{center}
\label{fig:varBoostedError}
\end{figure}

It is important to note that there  is no value for $\theta$
that yields reasonably low variance.  For example, Figure~\ref{fig:varBoostedError} shows the estimated relative 
error of $\hat\ell_\theta$ as a function of $\theta$ for $\gamma=45$ and all other parameters being the same
as in Table~\ref{tab:var boost}. The figure suggests that even if we knew the 
true variance-minimizing $\theta^*$ (obviating the need for approximating it), the estimator will  still not be useful. 

\subsection{Vanishing Relative Error Estimator}
\label{sec:ISVE}
Recognizing the deficiency of the variance boosted estimator \cite{asmussen2011efficient,blanchet2008efficient} propose the superior \emph{importance sampling vanishing relative eror} (ISVE) estimator. The main idea of the ISVE estimator is to split $\ell$
 into two parts:
\[
\ell=\bb P(M>\gamma)+\bb P(S>\gamma,M<\gamma),
\]
and estimate $\ell_1=\bb P(M>\gamma)$ and $\ell_2=\bb P(S>\gamma,M<\gamma)$ separately using two  different importance sampling schemes. In particular,
$\ell_1$ is estimated via
\begin{equation}
\label{ell1}
\hat\ell_1=\frac{\ell_\mathrm{as}(\gamma)}{\sum_{k=1}^d\bb I\{X_k>\gamma\}},\quad \v X\sim g(\v x),
\end{equation}
where $g$ is the mixture density:
\begin{equation}
\label{mixture}
g(\v x)\idef\frac{\phi_\Sigma(\v x-\v\nu)\sum_{k=1}^d \bb I\{x_k>\gamma\}}{\ell_\mathrm{as}(\gamma)},
\end{equation}
and the residual probability, $\ell_2$, is estimated via a  variance boosted estimator:
\begin{equation}
\label{residual}
\hat\ell_{2,\theta}(\gamma)=\frac{\exp(-\theta (\v Y-\v \nu)^\top\Sigma^{-1}(\v Y-\v\nu)/2)}{(1-\theta)^{d/2}}\bb I\{S>\gamma, M<\gamma\},\qquad \v Y\sim \bb P_\theta,
\end{equation}
where $\theta=1-\Theta(\ln^{-2}(\gamma))$. With this setup the ISVE estimator is
\[
\hat\ell_\mathrm{ISVE}=\hat\ell_1+\hat\ell_2
\]
 and it enjoys the vanishing relative error property \cite{asmussen2011efficient}:
\[
\frac{\var(\hat\ell_\mathrm{ISVE})}{\ell^2(\gamma)}\downarrow 0,\quad \gamma\uparrow\infty.
\]
Before we proceed to illustrate the practical performance of the ISVE estimator, we note that there are two issues that may indicate problematic performance.

First, in practical simulations one estimates the precision of an estimator $\hat l$ of $\ell$ by generating $n$ independent realizations, namely  $\hat l_{1},\ldots, \hat l_{n}$,    and then computing  the corresponding sample variance
$S_n^2=\frac{1}{n}\sum_{i=1}^n(\hat l_{i}-\bar l)^2,$ where
 $\bar l_n=(\hat l_{1}+\cdots+\hat l_{n})/n$. Ideally,  $S_n/(\bar l_n\sqrt{n})$ would then yield a consistent estimator of the relative error of the sample mean and in numerical experiments we would report either the pair $\bar l_n$ and $S_n/(\bar l_n\sqrt{n})$, or (say) the 95\% approximate confidence interval $\bar l_n\pm 1.96\times S_n/\sqrt{n}$. Unfortunately, these simple precision estimates cannot be applied to the
 $\hat\ell_\mathrm{ISVE}$ estimator, because the sample variance of $n$ independent replications of \eqref{ell1} 
is not a robust estimator of the true variance of $\hat\ell_1$. This is formalized in the following result, proved in the Appendix.
\begin{lem}[Inefficiency of Sample Variance]\label{proposition}
Let $
S_n^2$
 be the sample variance based on $n$ independent replications of \eqref{ell1} . Then, $S_n^2$ is not a logarithmically efficient estimator:
 \[
\limsup_{\gamma\uparrow\infty}\frac{\ln\var(S_n^2)}{\ln\var(\hat\ell_1)}<2,
\]
so that the relative error in estimating the precision of $\hat\ell_\mathrm{ISVE}$ grows at an exponential rate in $\ln^2(\gamma)$.
\end{lem}

One practical consequence of the result above is that the relative error of $\hat\ell_\mathrm{ISVE}$ is severely underestimated during simulation, and frequently reported as being zero. In contrast to this negative result for the ISVE estimator, in Corollary~\ref{cor:var} we show that our new estimator  enjoys an asymptotically efficient estimator of its true variability. Even better,  estimation of the
precision of our estimator is not more difficult than estimating $\ell$ itself.

A second problematic issue is that, as we have already seen, the variance boosted estimator \eqref{var boost} is unreliable for 
estimating $\ell$, and that there is no value for $\theta$ that will render it a useful estimator. Since  \eqref{residual} only differs from
\eqref{var boost} with the addition of the constraint $M<\gamma$ and in the  different choice of $\theta$, we  should not be surprised to find  that \eqref{residual} is also a poor estimator of  $\ell_2=\bb P(S>\gamma,M<\gamma)$. Indeed,  there is no good value for $\theta$ that can make the relative error of  \eqref{residual} small. The behavior of the relative error of $\hat\ell_{2,\theta}$ as a function of $\theta$ is qualitatively the same as  that on 
Figure~\ref{fig:varBoostedError}. 


\subsubsection{Quality of Asymptotic Approximation}
\label{sec:asymp}
One of the arguments in favor of the ISVE estimator is that, while $\hat\ell_2$
may be a noisy estimator, it is a very small second order residual term, and  will not affect noticeably the high accuracy of the leading order term $\hat\ell_1$. 

Unfortunately, unless one considers extremely small probabilities,  the leading contribution term of $\ell=\ell_1+\ell_2$ is not $\ell_1$, but the residual $\ell_2$. This does not contradict the fact that asymptotically $\ell\simeq \ell_1$, because in the presence of a positive correlation $\ell_1$, or equivalently $\ell_\mathrm{as}$, can be an extremely poor approximation to $\ell$.

\begin{table}[H]

\caption{Illustration of the inaccuracy of the asymptotic approximation for moderate values of $\ell(\gamma)$. The asymptotic approximation $\ell_\mathrm{as}$ becomes useful only for probabilities smaller than about  $10^{-233}$.}
\begin{center}
\begin{tabular}{c|c|c|c}  
$\gamma$ & $\ell_\mathrm{as}$ & 95\% CI for $\hat\ell$ & $(\hat\ell-\ell_\mathrm{as})/\ell_\mathrm{as}$\\
\hline
$15$ & $ 1.2113\ldots\times 10^{-26}$ & $0.012\pm 0.001$ & $1.0\times 10^{24}$\\
$20$ & $2.1830\ldots\times 10^{-32}$ & $(5.80\pm 0.013)\times 10^{-5}$ & $2.66\times 10^{27}$\\
$40$ & $ 1.4175\ldots\times 10^{-48}$ & $(6.33\pm 0.016)\times 10^{-15}$ & $4.5\times 10^{33}$\\
$60$ & $ 1.3872\ldots\times 10^{-59}$ & $(1.10\pm 0.017)\times 10^{-23}$ & $8\times 10^{35}$\\
$100$ & $ 4.4834\ldots\times 10^{-75}$ & $(8.04\pm .018)\times 10^{-38}$ & $1.8\times 10^{37}$\\
$500$ & $ 1.0481\ldots\times 10^{-135}$ & $( 3.39\pm .02)\times 10^{-105}$ & $3.2\times 10^{30}$\\
$1000$ & $ 2.3594\ldots\times 10^{-167}$ & $(6.94\pm .02)\times 10^{-145}$ & $3\times 10^{22}$\\
$1500$ & $  2.0634\ldots\times 10^{-187}$ & $(4.04\pm .03)\times 10^{-171}$ & $2\times 10^{16}$\\
$2500$ & $   2.6294\ldots\times 10^{-214}$ & $(2.94\pm .04)\times 10^{-207}$ & $1.1\times 10^7$\\
$3500$ & $   5.1912\ldots\times 10^{-233}$ & $(5.45\pm .04)\times 10^{-233}$ & $0.05$
\end{tabular}
\end{center}
\label{tab:asymp}
\end{table}

To 
illustrate the quality of the asymptotic approximation take the 
instances in Table~\ref{tab:asymp}, where $d=10$ and $\v\nu=\v 0, \Sigma= 0.25^2\times(0.9\times\v 1\v 1^\top+(1-0.9)\times\m I)$. The table shows the asymptotic value $\ell_\mathrm{as}$ for different $\gamma$ (second column), together with its relative deviation from the true $\ell$ (last column).
The table also displays  $\hat\ell$ with its approximate 95\% confidence interval  based on $n=10^6$ independent replications of our method in Section~\ref{sec:new}.

We conclude from Table~\ref{tab:asymp}  that the asymptotic approximation is useless for moderate values of $\gamma$
(deviating from the true value of $\ell$ by as much as $10^{37}$), and only becomes useful for extremely small probabilities (smaller than $10^{-233}$).


In fact, for the numerical experiment above it can be shown (see \cite{kortschak2014second}) that the \emph{second order} asymptotic term is:
\[
\begin{split}
\ell(\gamma)-\ell_\mathrm{as}(\gamma)&=\ell(\gamma)-d\overline\Phi(\ln\gamma)\simeq d (d-1)\exp((1-\rho^2)/2)\frac{\ln(\gamma)}{\gamma^{1-\rho}}\overline\Phi(\ln\gamma)
\end{split}
\]
In other words, the relative error of the asymptotic approximation,
\[
\frac{\ell(\gamma)-\ell_\mathrm{as}(\gamma)}{\ell_\mathrm{as}(\gamma)}=\Theta\left(\frac{\ln\gamma}{\gamma^{1-\rho}}\right),
\]
decays at a polynomial rate, and this rate can be extremely slow when $\rho$ is close to unity. Even worse, with $\rho=0.99$ the term $\frac{\ln\gamma}{\gamma^{1-\rho}}$ \emph{increases} as a function of $\gamma$ for values up to $\gamma\leq \exp(1/(1-\rho))=\exp(100)\approx 10^{43}$, and only starts decaying  to zero for $\gamma>\exp(100)$.

%
%
 %

\subsection{Exponentially Tilted Estimator}
\label{sec:new}
Given the failure of  the estimators described above, a natural question arises. What kind of estimator will succeed in being  both theoretically efficient and 
exhibit low variance in practical simulations? 

To answer this question we start by examining 
the quite natural proposal of Gulisashvili  and Tankov (GT)  \cite[Equation (70)]{gulisashvili2016tail}, which  
 can be written as follows
: 
\begin{equation}
\label{GT right tail}
\hat\ell_\mathrm{GT}\idef\exp\left(\frac{\v\mu^\top\Sigma^{-1}\v\mu}{2}-\v\mu^\top \Sigma^{-1}(\v Y-\v\nu)\right)\bb I\{S\geq \gamma\}, \qquad \v Y\sim  \mathsf{N}(\v \mu+\v\nu,\Sigma), 
\end{equation}
where the parameter $\v\mu$ is chosen by minimizing an asymptotic approximation to the second moment \cite[Equation (71)]{gulisashvili2016tail}.

Unfortunately, \eqref{GT right tail} also performs  poorly, just like the estimators  in the last section.\footnote{We remark that the GT estimator applies to the more general setting of  sums \emph{and differences} of log-normals. This  generality of the GT estimator, however, comes at the cost of not being the most efficient estimator for sums --- the case we consider here.} 
 This poor practical performance is compounded by the fact that  there is no proof of the asymptotic efficiency of 
\eqref{GT right tail} as $\gamma\uparrow\infty$ (see \cite[Page 40]{gulisashvili2016tail}). 

The reason why the  estimator \eqref{GT right tail} performs poorly is that it uses a \emph{single} exponential tilting parameter $\v\mu$, which is insufficient to induce the mutually-exclusive mode of occurrence of the rare-event: $\bb P(X_k=M|S>\gamma)\simeq \frac{\bb P(X_k>\gamma)}{\ell}$ (see part 1 of Lemma~\ref{lem}). In other words, with asymptotic probability $\bb P(X_k>M)/\ell$, each $X_k$ is the maximal term that causes the sum  to up-cross $\gamma$, and the single exponential tilting parameter $\v\mu$ in \eqref{GT right tail} cannot account for this mutually-exclusive behavior.

 Instead, to obtain a provably efficient estimator with excellent practical and theoretical performance, we must
introduce $d$ distinct exponential tilting parameter vectors $\v\mu_1,\ldots,\v\mu_d$, where  each  $\v\mu_k$ is tasked to deal with the event $\{S>\gamma,X_k=M\}$. The new set of $d$ tilting parameters are also determined using an error estimate different from the one used in \eqref{GT right tail} when we have a single tilting parameter. Thus, our proposal uses an estimator of the form \eqref{GT right tail} for each term, $\hbar_i(\gamma)$, in the decomposition:
\[
\ell(\gamma)=\sum_{i=1}^d \underbrace{\bb P(S>\gamma, X_i=M)}_{\hbar_i(\gamma)}.
\]
%
The estimator of  each $\hbar_k$ based on one replication is 
\begin{equation}
\label{new est}
\hat\hbar_k(\gamma)=\exp\left(\frac{\v\mu_k^\top\Sigma^{-1}\v\mu_k}{2}-\v\mu_k^\top \Sigma^{-1}(\v Y-\v \nu)\right)\bb I\{S>\gamma, X_k=M\},
\end{equation}
where under the measure $\bb P_{\v\mu_k}$ with expectation
$\bb E_{\v\mu_k}$ we have  $\ln \v X=\v Y\sim \mathsf{N}(\v\nu+\v\mu_k,\Sigma)$, and 
$\v\mu_k$ is  the solution to the non-linear optimization:
\begin{equation}
\label{mu}
\begin{split}
\min_{\v\mu}\;\frac{1}{2}\v\mu^\top\Sigma^{-1}\v\mu,\quad \textrm{subject to:}&\\
 \textstyle g_1(\v\mu)=\exp(\mu_k+\nu_k)+\sum_{i\not =k}\exp(\mu_i+\nu_i+\frac{\sigma_i^2}{2})-\gamma&\geq 0\\
\textstyle  g_2(\v\mu)=\mu_k+\nu_k+\frac{\sigma_k^2}{2}-\max_{j\not =k}\{\mu_j+\nu_j+\frac{\sigma_j^2}{2}\}&\geq 0
\end{split}
\end{equation}
To construct the overall estimator $\hat\ell$ of $\ell$, we can use stratification with a total computing budget of $n=n_1+\cdots+n_d$ replications, whereby we allocate $n_k$ samples  to estimate each $\hbar_k$ independently. We then take the sum of the estimators of $\hbar_k$'s  as our stratified estimator of $\ell$. In other words,
\begin{equation}
\label{stratified est}
\hat\ell(\gamma)=\sum_{k=1}^d\frac{1}{n_k}\sum_{j=1}^{n_k}\hat\hbar_{k,j}(\gamma),
\end{equation}
where $\sum_k n_k=n$ and $\hat\hbar_{k,1},\ldots,\hat\hbar_{k,n_k}$ are iid copies of \eqref{new est}. We then  have the following efficiency result.
\begin{thm}[Logarithmic Efficiency]
Suppose we select the
stratified allocation such that $n_i\varpropto n\times \bb P(X_i>\gamma)$. 
Then, the estimator \eqref{stratified est} is  unbiased and logarithmically efficient with relative error $\var(\hat\ell(\gamma))/\ell^2(\gamma)=\c O(\ln\gamma)$ as $\gamma\uparrow\infty$.
\end{thm}

\begin{proof} First note that  choosing  
$n_k=n\times \bb P(X_k>\gamma)/\ell_\mathrm{as}$ satisfies the constraint $n=\sum_k n_k$, but  is in conflict with
the constraint that the $n_k$'s have to be integers.
One simple  solution is to simply round up to the nearest integer,  and violate the constraint $n=\sum_k n_k$. For large enough $n$, the residual $n-\sum_k n_k$ will be negligible.  
Another solution, which we adopt in our computer implementation, is to use  a widely-used \emph{randomized}  stratification scheme, as described in, for example, \cite[Algorithm 14.2]{kroese2011handbook}.  
 
Next, with the above allocation for each $n_k$, the variance of the  stratified estimator \eqref{stratified est} is:
\begin{equation*}
\var(\hat\ell)=\frac{1}{n}\sum_{k=1}^d \frac{n}{n_k} \var(\hat\hbar_k)=\frac{\ell_\mathrm{as}}{n}\sum_{k=1}^d \frac{\var(\hat\hbar_k)}{\bb P(X_k>\gamma)}
\end{equation*}
Therefore, using the result 3. in Lemma~\ref{lem},  the relative error of $\hat\ell$ as $\gamma\uparrow\infty$ is 
\begin{equation}
\begin{split}
\frac{n\var(\hat\ell(\gamma))}{\ell^2(\gamma)}&\simeq\frac{n\var(\hat\ell(\gamma))}{\ell^2_\mathrm{as}(\gamma)}=\frac{1}{\ell_\mathrm{as}(\gamma)}\sum_{k=1}^d\bb P(X_k>\gamma)\frac{\var(\hat\hbar_k)}{[\bb P(X_k>\gamma)]^2}=\c O (\ln\gamma).
\end{split}
\end{equation}
The last equation  shows that $\frac{\ln\var(\hat\ell^2(\gamma))}{\ln\ell(\gamma)}\rightarrow 2$ as $\gamma\uparrow\infty$.
\end{proof}
We can now see that the rate of growth of the relative error of our estimator, namely $\c O(\ln\gamma)$, is significantly slower than
 the rate of growth of the variance boosted estimator, $\c O([\ln\gamma]^{d/2+1}\gamma^{1/4})$.

%
%
%
%
%

Since the proof of the following lemma is  long, it is delegated to the appendix.
\begin{lem}[Asymptotics for $\hat\hbar_k$]
\label{lem} As $\gamma\uparrow\infty$, we have that:
\begin{enumerate}
\item
$
\bb E_{\v\mu}[\hat\hbar_k(\gamma)]=\bb P(S>\gamma,X_k=M)\simeq \bb P(X_k>\gamma).
$
\item The asymptotic solution to \eqref{mu} is
\[
\v\mu^*=\frac{\ln(\gamma)-\nu_k}{\sigma_k^2}\Sigma\v e_k,
\]
where $\v e_k$ is the unit vector with $1$ in the $k$-th position.
\item We have $\frac{\var(\hat\hbar_k)}{[\bb P(X_k>\gamma)]^2}=\c O(\ln\gamma)$, and with $\v\mu$ solving \eqref{mu} the $(m+1)$-st moment satisfies:
\begin{equation}
\label{mth moment}
 \bb E_{\v\mu} \hat\hbar^{m+1}_k=\Theta(\ln^m(\gamma)\hbar_k^{m+1})
\end{equation}
\end{enumerate}
\end{lem}
 Note that part 1. of the above lemma immediately  yields the following corollary, which was originally proved in \cite{asmussen2008asymptotics} using 
a  different argument. 
\begin{corollary}[Right-Tail Asymptotics] $\ell(\gamma)\simeq \sum_{i=1}^d\bb P(X_k>\gamma)$ as $\gamma\uparrow\infty$.
\label{cor:asymp}
\end{corollary}
More importantly, part 3. of  Lemma~\ref{lem} gives us a robustness guarantee that is not enjoyed by any of the competing estimators. 

\begin{corollary}[Logarithmically Efficient Variance Estimator]
\label{cor:var}
Let $
S_{n_k}^2$
 be the sample variance based on $n_k$ independent replications of \eqref{new est}. Then, $S_n^2$ is a logarithmically efficient estimator:
 \[
\liminf_{\gamma\uparrow\infty} \frac{\ln \var(S_{n_k}^2)}{\ln\var(\hat\hbar_k)}=2,
\]
where the rate of growth is  $\frac{\var(S_{n_k}^2)}{\var^2(\hat\hbar_k)}=\c O(\ln\gamma)$.
\end{corollary}
\begin{proof} 
 Using \eqref{mth moment}, consider the following calculations:
\[
\begin{split}
\frac{n_k\var(S_{n_k}^2)}{\var^2(\hat\hbar_k)}&=\frac{\bb E_{\v\mu}(\hat\hbar_k(\gamma)-\hbar_k(\gamma))^4}{[\bb E_{\v\mu}(\hat\hbar_k(\gamma)-\hbar_k(\gamma))^2]^2}-1+\frac{2}{n_k-1}
\\
&=\frac{\Theta(\ln^3(\gamma)\hbar_k^{4})+
\hbar_k^{4}+\hbar_k^{2}\Theta(\ln(\gamma)\hbar_k^{2})-
\hbar_k\Theta(\ln^2(\gamma)\hbar_k^3)-
4\hbar^4_k}{[\Theta(\ln(\gamma)\hbar_k^2)-\hbar_k^2]^2}-1+\frac{2}{n_k-1}\\
&=\frac{\Theta(\ln^3\gamma)+\Theta(\ln\gamma)-\Theta(\ln^2\gamma)-
3}{[\Theta(\ln\gamma)-1]^2}-1+\frac{2}{n_k-1}\\
&=\Theta(\ln(\gamma))-1+\frac{2}{n_k-1}
\end{split}
\]
\end{proof}
Therefore, a major advantage of our proposed estimator \eqref{stratified est} is that 
estimating its variance  is asymptotically not more difficult than estimating  $\ell$ itself. 

In retrospect, we can see that  the excellent theoretical properties of our estimator are due mainly
 to the breaking of the symmetry in the sum 
$S=X_1+\cdots+X_d$ by distinguishing each and every  $X_i$ as the overall maximum.  In contrast, the poorly performing estimators \eqref{GT right tail} and \eqref{residual} (and hence $\hat\ell_\mathrm{ISVE}$)
both induce a simple change of measure that does not conform to the mutually-exclusive asymptotic behavior of $\bb P(S>\gamma,X_k=M)\simeq\bb P(X_k>\gamma),\;k=1,\ldots,d$.

Finally, we remark on the unusual way of selecting $\v\mu$ via the optimization \eqref{mu}. Why do we not simply use the asymptotic approximation $\v\mu^*$ in
Lemma~\ref{lem}? The answer is that, while asymptotically the matrix $\Sigma$ is irrelevant, it is still relevant for very large values of $\gamma$, and our change of measure should reflect this dependence.
The asymptotic solution $\v\mu^*$ does not reflect this dependence. 
Thus, \eqref{mu} was designed with two objectives in mind:
(1)  good practical performance for finite $\gamma<\infty$, where the full $\Sigma$ is relevant; (2)  asymptotic optimality as $\gamma\uparrow\infty$, where $\Sigma$ is irrelevant. The optimization program \eqref{mu} transitions from objective (1) to objective (2)  in a continuous way.

\subsection{Numerical Comparison}
In this section we show that the superior theoretical properties of \eqref{stratified est} convert into
excellent practical performance. In fact, the numerical simulations suggest that our proposed estimator is the only one capable of estimating $\ell$ in many settings  of practical interest. 

\subsubsection{Comparison with ISVE estimator}
  Consider estimating $\ell(\gamma)$ with
	\[
	d=30,\rho=0.9,\;\v\nu=\v 0,\;\Sigma= 0.25^2\times(\rho\times\v 1\v 1^\top+(1-\rho)\times\m I)
	\]
and different values of  $\gamma$. Table~\ref{tab:ISVE compare} gives the results using $n=10^6$ replications.  For the ISVE estimator we attempted to optimize the performance of the estimator by manually selecting the best possible $\theta$. Our choice for this tuning parameter 
is thus given in brackets in the third column.

\begin{table}[H]
\caption{Comparative performance of the ISVE and exponentially tilted estimators with $\rho=0.9$.}
\label{tab:ISVE compare}
\begin{center}
\begin{tabular}{c|c|c|c|c|c|c}  
\multicolumn{3}{c}{} & \multicolumn{2}{|c}{relative error \%}& \multicolumn{2}{|c}{work normalized rel. var.}  \\
 \hline
$\gamma$& $\hat\ell$  & $\hat\ell_\mathrm{ISVE}$  & $\mathrm{RE}(\hat\ell)$ & $\mathrm{RE}(\hat\ell_\mathrm{ISVE})$ &  $\mathrm{WNRV}(\hat\ell)$ & $\mathrm{WNRV}(\hat\ell_\mathrm{ISVE})$  \\
\hline
40& 0.116 & 0.114 ($\theta=0.5$) &0.63 & 2.0 & 0.00032 & 0.00080 \\
100& $2.17\times 10^{-7}$ & $1.18\times 10^{-7}$  ($\theta=0.6$) &0.98 & 40&0.00061 & 0.31 \\ 150& $6.83\times 10^{-12}$ & $5.75\times 10^{-13}$  ($\theta=0.75$) &1.1 & 84&0.00093 & 1.12 \\
200& $ 7.75\times 10^{-16}$ & $2.09\times 10^{-17}$  ($\theta=0.8$) &1.2 & 95& 0.0010 & 1.22 \\
400& $6.57\times 10^{-28}$ & $3.08\times 10^{-39}$  ($\theta=0.9$) &1.4 & 80& 0.0011 & 1.34\\
$10^3$& $ 1.61\times 10^{-49}$ & $1.21\times 10^{-80}$  ($\theta=0.95$) &1.7 & 100& 0.002 & 2.02\\
$10^4$& $3.60\times 10^{-132}$ & $1.80\times 10^{-294}$  ($\theta=?$) &2.1 & -&  0.0024 & -\\
\end{tabular}
\end{center}
\end{table}
A number of conclusions can be drawn from the table.

First,  the ISVE estimator does not have acceptably low variance for both small $\gamma$ (when the event is not rare) and  for  large $\gamma$ (when the event is rare). 

Second, as with Figure~\ref{fig:varBoostedError}, any attempt to optimize with respect to $\theta$ is fruitless, because there appears to be no value for $\theta\in [0,1)$ that yields low variance. 

Third, in the last row of the table, it was not possible to induce the event
$\{S>\gamma,M<\gamma\}$ no matter what the value of $\theta$. In other words, $\{S>\gamma,M<\gamma\}$ remains a rare-event for all values of $\theta\in [0,1)$, and with very high probability $\hat\ell_\mathrm{ISVE}=\hat\ell_1+\hat\ell_2=\hat\ell_1$. Thus, despite the vanishing relative error property of the ISVE estimator, its   performance deteriorates as $\gamma$ becomes smaller and smaller to the point that it does not deliver meaningful estimates.

The next example in Table~\ref{tab:right tail}  suggests that our estimator remains robust even 
in very high dimensions of up to $d=60$.


\begin{table}[H]
\centering
\caption{Performance for $d=60$, $n=10^6$, $\v \nu = \v 0$,
$\Sigma=0.5\times\v 1\v 1^\top +0.5\times\m I$.} 
\begin{tabular}{c|c|c|c|c|c|c}
\multicolumn{3}{c}{} & \multicolumn{2}{|c}{relative error \%}& \multicolumn{2}{|c}{work normalized rel. var.}  \\
 \hline
 $\gamma$ & $\widehat{\ell}$ &   $\widehat{\ell}_{\rm ISVE}$&${\rm RE}(\widehat{\ell})$&${\rm RE}(\widehat{\ell}_{\rm ISVE})$& {WNRV}$(\widehat{\ell})$ & WNRV$(\widehat{\ell}_{\rm ISVE})$ \\ \hline 
600 & \num{1.98E-03} & \num{5.77E-07} & \num{0.837} & \num{51.6} & \num{1.02E-03} & \num{4.88e+00}  \\
900 & \num{2.81E-04} & \num{4.00E-10} & \num{0.893} & \num{15.4} & \num{1.20E-03} & \num{ 4.42e-01}  \\
1200 & \num{5.91E-05} & \num{4.16E-11} & \num{0.93} & \num{3.16} & \num{1.36E-03} & \num{ 1.75e-02}    \\ 
1500 & \num{1.57E-05} & \num{7.92E-12} & \num{0.964} & \num{1.23} & \num{1.52E-03} & \num{ 2.56e-03}  \\
1800 & \num{5.01E-06} & \num{2.01E-12} & \num{0.987} & \num{1.50} & \num{1.57E-03} & \num{ 3.41e-03}    \\ 
2100 & \num{1.79E-06} & \num{6.05E-13} & \num{1.012} & \num{0.0543} & \num{2.03E-03} & \num{ 5.10e-06} \\
2400 & \num{7.18E-07} & \num{2.12E-13} & \num{1.029} & \num{2.84E-03} & \num{1.56E-03} & \num{ 1.11e-08} \\
2700 & \num{3.08E-07} & \num{8.30E-14} & \num{1.046} & \num{8.93E-04} & \num{1.63E-03} & \num{ 1.15e-09}  \\ 
3000 & \num{1.44E-07} & \num{3.54E-14} & \num{1.057} & \num{6.82E-04} & \num{2.02E-03} & \num{ 5.97e-10}  \\
3300 & \num{7.02E-08} & \num{1.63E-14} & \num{1.069} & \num{8.30E-04} & \num{2.05E-03} & \num{9.60E-10}
\end{tabular}
\label{tab:right tail}
\end{table}

\subsubsection{Comparison with Modified Asmussen-Kroese estimator}
In addition to the ISVE estimator,   the \emph{modified Asmussen-Kroese}  (MAK) estimator \cite[Equation 3.6]{kortschak2013efficient} also enjoys the  vanishing relative error property.
In comparing 
 \eqref{stratified est} with the MAK estimator, we make the following observations.

First, the MAK estimator requires the solution of a non-linear equation for
every replication. This aspect of the estimator poses nontrivial problems: (1) sometimes no solution exists; (2)  Newton's method may take many iterations to converge, making  the running time of the estimator large.

Second, while our estimator
 \eqref{stratified est} was shown to be second-order efficient, ensuring  reliable estimation of its precision, no such efficiency result is provided for the MAK estimator, and
in numerical experiments we sometimes observed significant underestimation of the true variance of the MAK estimator. 

Third, observe that the MAK estimator reduces to
the Asmussen-Kroese (AK) estimator \eqref{AKest} in the independent case:  
$\v\nu=\nu\v 1, \Sigma=\sigma^2\m I$. Table~\ref{tab:var boost}  shows that when   $\sigma$ is small
 our estimator can outperform 
the (modified) Asmussen-Kroese estimator by orders of magnitude.
For example, note how for $\gamma\in[42,51]$ the AK estimator severely underestimates the true probability by as much as an order of $10^{-4}$.
Interestingly,   the 
 Asmussen-Kroese estimator has superior and unrivaled performance only in cases with larger $\sigma$, say $\sigma\geq 1$.

Finally, Table~\ref{tab:MAK} below confirms that the MAK estimator inherits the  poor performance of the Asmussen-Kroese estimator 
for small $\sigma$. In this particular example we use
 \[
d=10,\;
\v\nu=\v 0,\;\rho=0.2,\;\Sigma=0.25^2(\rho\v 1\v 1^\top+(1-\rho)\m I).
\]
Observe how for $\gamma\in[26,30]$ the MAK estimator underestimates the true probability by as much as an order of $10^{-3}$.
\begin{table}[H]
\caption{Comparative performance for $d=10,n=10^6,
\v\nu=\v 0,\rho=0.2,\Sigma=0.25^2(\rho\v 1\v 1^\top+(1-\rho)\m I)$.}
\label{tab:MAK}
\begin{center}
\begin{tabular}{c|c|c|c|c|c|c}  
\multicolumn{3}{c}{} & \multicolumn{2}{|c}{relative error \%}& \multicolumn{2}{|c}{work normalized rel. var.}  \\
 \hline
  $\gamma$ &  $\hat\ell_\mathrm{MAK}$& $\hat\ell$ & $\mathrm{RE}(\hat\ell_\mathrm{MAK})$  & $\mathrm{RE}(\hat\ell)$ &   $\mathrm{WNRV}(\hat\ell_\mathrm{MAK})$&   $\mathrm{WNRV}(\hat\ell)$ \\
\hline
					%
					  15 &   0.00195  &  0.00198  &    0.420  &   0.669   & 0.00531  & \num{4.15e-05}\\
           16  & 0.000373   &0.000370   &   0.660    &  0.724    & 0.0130   &\num{5.03e-05}\\
           17  & \num{6.47e-05}   &\num{6.47e-05}   &    1.07    &  0.775    & 0.0349   &\num{5.20e-05}\\
           18  & \num{1.00e-05}   &\num{1.02e-05}   &    1.80    &  0.823    &  0.096   &0.000102\\
           19  & \num{1.57e-06}   &\num{1.52e-06}   &    3.10    &   0.87    &   0.28   &\num{6.71e-05}\\
           20  & \num{2.02e-07}   &\num{2.15e-07}   &    5.60    &  0.937    &  0.941   &\num{8.09e-05}\\
           21  & \num{3.11e-08}   &\num{2.99e-08}   &     9.2    &   1.00    &   2.56   &0.000166\\
           22  & \num{3.80e-09}   &\num{3.91e-09}   &    15.9    &   1.06    &   7.58   &0.000108\\
           23  & \num{3.22e-10}   &\num{5.15e-10}   &    19.0    &   1.07    &   11.2   &0.000327\\
           24  & \num{5.63e-11}   &\num{6.61e-11}   &    44.7   &   1.14    &   61.57   &0.000123\\
           25  &  \num{6.09e-12}   &\num{8.42e-12}   &    42.0    &   1.18    &   55.0   & 0.00036\\
           26  & \num{4.63e-13}   &\num{1.05e-12}   &    71.6    &   1.23    &   159   &0.000197\\
           27  & \num{1.90e-14}   &\num{1.33e-13}   &    31.5   &   1.40   &   30.0  &0.000176\\
           28  & \num{4.85e-15}   &\num{1.69e-14}   &    60.0    &   1.49    &   110.   &0.000187\\
           29  & \num{9.12e-17}   &\num{2.15e-15}   &    60.4    &    1.5    &   113.34   &0.00021378\\
           30  & \num{6.37e-18}   &\num{2.74e-16}   &    53.1    &   1.54    &   87.585   &0.00018874
\end{tabular}
\end{center}
\end{table}

\section{Conclusions}

We have presented new methodology for the accurate estimation
of the cdf, pdf, and tails of the SLN distribution. 
In all three cases the new methodology yields  estimators that can tackle parameter settings which are currently intractable  with existing methods. 
In the cdf and pdf cases, the proposed estimators permit  additional variance 
reduction via Quasi Monte Carlo.
In the right-tail case, the new exponentially tilted estimator is shown to be, not only logarithmically or weakly efficient, but also second-order efficient.
This permits us to have greater confidence in all error estimates derived from simulation. 

 One of the crucial observations we can draw from a number of numerical 
experiments is that sometimes a strongly efficient estimator ($\hat\ell_\mathrm{ISVE},\hat\ell_\mathrm{MAK}$) may not 
necessarily exhibit low variance in practical simulations, and may be bettered
by a much simpler weakly efficient estimator.

In subsequent work, we will explore using the sequential sampling ideas in Section~\ref{sec:cdf} for the estimation of the distribution of an iid sum of random variables with an arbitrary (be it light- or heavy-tailed) distribution.






\appendix
\section*{Appendix}

\subsection*{Proof of Theorem~\ref{theorem:VRE left}}
\begin{proof}
Under the assumption
that $\Sigma_{11}<\Sigma_{1j}$ for $j\not =1$, we have $\ell(\gamma)\simeq \bb P(X_1<\gamma)$, see \cite{gulisashvili2016tail} for a proof. Therefore, using the upper bound
\[
\bb E_{\v 0}\hat\ell^2_{\v 0}(\gamma)=\bb E_{\v 0} \prod_{j=1}^d \Phi^2(\alpha_j(Z_1,\ldots,Z_{j-1}))\leq \bb E \Phi^2(\alpha_1)=[\bb P(X_1<\gamma)]^2\simeq \ell^2(\gamma),
\]
we have as $\gamma\downarrow 0$,
\[
\frac{\var_{\v 0}(\hat\ell_{\v 0})}{\ell^2}= \frac{\bb E_{\v 0}\hat\ell^2_{\v 0}(\gamma)}{\ell^2(\gamma)}-1\leq \frac{[\bb P(X_1<\gamma)]^2}{\ell^2(\gamma)}-1\rightarrow 0
\]
\end{proof}

\subsection*{Proof of Theorem~\ref{theorem:left tail}}

\begin{proof}
 To proceed with the proof we recall the following three facts. 
First, note that
$\ell(\gamma)=\bb P(\v 1^\top\exp(\v Y)\leq \gamma)$, where 
 $\v Y=\v\nu+\m L\v Z$. 
Using Jensen's inequality, we have that for any $\v w\in \c W$:
\begin{equation}
\label{Jensen}
\begin{split}
\ell(\gamma)&=\textstyle\bb P( \v w^\top\exp(\v Y-\ln \v w) \leq \gamma)\leq \textstyle\bb P( \v w^\top\ln(\v w)-\v w^\top\v Y  \geq -\ln\gamma)\\
&\leq\textstyle \overline\Phi\left( \frac{\v w^\top\v\nu-\ln\gamma-\v w^\top\ln \v w}{\sqrt{\v w^\top\Sigma\v w}}\right)\leq  \exp\left(-\frac{(\v w^\top\v\nu-\ln\gamma-\v w^\top\ln \v w)^2}{2\v w^\top\Sigma\v w}\right)
\end{split}
\end{equation}
Second,  denote 
$\bar{\v w}=\argmin_{\v w\in \c W} \v w^\top \Sigma \v w$ and
 the set  $\c C_\gamma\equiv\{\v z: \v 1^\top \exp(\m L \v z+\v \nu)\leq \gamma\}$.
Then, we have the asymptotic formula, proved in \cite[Formulas (13) and (63)]{gulisashvili2016tail}:
\begin{equation}
\label{asymp1}
\ln\ell^2(\gamma)\simeq c_1-\frac{(\ln(\gamma)-\bar{\v w}^\top\v \nu+\bar{\v w}^\top\ln\bar{\v w})^2}{\bar{\v w}^\top\Sigma\bar{\v w}}-(1+d)\ln(-\ln\gamma),\qquad \gamma\downarrow 0,
\end{equation}
where $c_1$ is a constant, independent of $\gamma$.  
Thirdly, 
consider the nonlinear optimization
\begin{equation}
\label{mubarprog}
\bar{\v\mu}=\argmin_{\v\mu}\left\{\|\v\mu\|^2-\frac{(\ln(\gamma)-\bar{\v w}^\top(\v \nu-\m L\v\mu)+\bar{\v w}^\top\ln\bar{\v w})^2}{2\bar{\v w}^\top\Sigma\bar{\v w}}\right\}
\end{equation}
with explicit solution 
\begin{equation}\label{mubar}
\bar{\v\mu}=\frac{\ln\gamma-\bar{\v w}^\top\v \nu+\bar{\v w}^\top\ln\bar{\v w}}{\bar{\v w}^\top\Sigma\bar{\v w}}\;\m L^\top \bar{\v w}
\end{equation}
Then,  we obtain the following bound on the second moment:
\[
\begin{split}
\bb E_{\v\mu^*}\hat\ell^2(\gamma)=\bb E_{\v\mu^*} \exp(2\psi(\v Z;\v\mu^*))&=
\bb E \exp(\psi(\v Z;\v\mu^*))\bb I\{\v Z\in \c C_\gamma\}\\
&=\textstyle\bb E \exp(\|\v\mu^*\|^2_2) \bb I\{(\v Z-\v\mu^*)\in \c C_\gamma\} \prod_j\overline\Phi(\mu_j^*-\alpha_j(\v Z-\v\mu^*))\\
&\leq \textstyle \exp(\|\v\mu^*\|^2_2)\bb P((\v Z-\v\mu^*)\in \c C_\gamma)\\
\textrm{using \eqref{Jensen} }&\leq \textstyle \exp(\|\v\mu^*\|^2_2) \overline\Phi\left( \frac{(\v\nu-\m L\v\mu^*)^\top\v w^*-\ln\gamma-(\v w^*)^\top\ln \v w^*}{\sqrt{(\v w^*)^\top\Sigma\v w^*}}\right) \\
\textrm{via \eqref{Jensen}+\eqref{mubarprog} }&\leq  \textstyle \exp\Big(\|\bar{\v\mu}\|^2_2
-\frac{(\bar{\v w}^\top(\v\nu-\m L\bar{\v\mu})-\ln\gamma-\bar{\v w}^\top\ln \bar{\v w})^2}{2\bar{\v w}^\top\Sigma\bar{\v w}}
\Big)
\end{split}
\]
By
substituting \eqref{mubar} in the last line, we obtain the upper bound
\[
\bb E_{\v\mu^*}\hat\ell^2\leq  \textstyle\exp\left(-\frac{(\ln\gamma-\bar{\v w}^\top\v \nu+\bar{\v w}^\top\ln\bar{\v w})^2}{\bar{\v w}^\top\Sigma\bar{\v w}}\right)
\]
In other words, from \eqref{asymp1} we deduce that
\[
\frac{\bb E_{\v\mu^*}\hat\ell^2(\gamma)}{\ell^2(\gamma)}= \c O( (-\ln\gamma)^{(d+1)}) ,\qquad \gamma\downarrow 0 
\]
and therefore
\[
\liminf_{\gamma\downarrow 0} \frac{\ln \bb E_{\v\mu^*}\hat\ell^2(\gamma)}{\ln\ell(\gamma)}=2,
\]
which implies that the algorithm is  logarithmically efficient with respect to  $\gamma$. 
\end{proof}

\subsection*{Proof of Theorem~\ref{proposition}}
\begin{proof}
Let  $N\idef \sum_{i=1}^d\bb I{\{X_i>\gamma\}},$ so that $\ell_1(\gamma)=\bb P(N\geq 1)\simeq \ell_\mathrm{as}$ and
 the residual 
\[\textstyle
r(\gamma)\idef\ell_\mathrm{as}-\ell_1(\gamma)=\sum_{i<j}\bb P(X_i>\gamma,X_j>\gamma)+o\left(\sum_{i<j}\bb P(X_i>\gamma,X_j>\gamma)\right).
\]
Note that $\bb P(N>1)=\Theta(r(\gamma))$ and $\bb P_g(N=1)=\bb P(N=1)/\ell_\mathrm{as}(\gamma)=\Theta( 1)$, where $g$ is the mixture density defined in \eqref{mixture}. We thus obtain 
\[
\begin{split}
\bb E_g\left|\hat\ell_1-\ell_1(\gamma)\right|^m&=
\sum_{j=1}^d\bb E_g \left[\left|\hat\ell_1-\ell_1(\gamma)\right|^m\,\bb I{\{N=j\}}\right]\\
&=|\ell_\mathrm{as}(\gamma)-\ell_1(\gamma)|^m\bb P_g(N=1)+\sum_{j=2}^d\left|\frac{\ell_\mathrm{as}(\gamma)}{j}-\ell_1(\gamma)\right|^m\bb P_g(N=j)\\
&=r^m(\gamma)\bb P_g(N=1)+\Theta(\ell_\mathrm{as}^m) \bb P_g(N>1)\\
&=r^m(\gamma)\bb P_g(N=1)+\Theta(\ell_\mathrm{as}^{m-1}) \bb P(N>1)\\
&=\Theta\left(r^m(\gamma)\right)+\Theta\left(\ell_\mathrm{as}^{m-1} r(\gamma)\right).
\end{split}
\]
Therefore, since $r(\gamma)=o(\ell_\mathrm{as}(\gamma))$,  we have:
\[
\begin{split}
n\var(S_n^2)&=\bb E_g(\hat\ell_1-\ell_1(\gamma))^4+\left(\frac{2}{n-1}-1\right)\var^2(\hat\ell_1)
\\
&=\Theta(r^4)+\Theta(\ell_\mathrm{as}^{3}(\gamma) r(\gamma))+\Theta(\var^2(\hat\ell_1))\\
&=\Theta(\ell_\mathrm{as}^{3}(\gamma) r(\gamma))+\Theta(\var^2(\hat\ell_1)),
\end{split}
\]
and
\[
\var^2(\hat\ell_1)=\Theta(r^4)+\Theta(\ell_\mathrm{as}(\gamma) r^3(\gamma))+\Theta(\ell_\mathrm{as}^{2}(\gamma)r^2(\gamma))=\Theta(\ell_\mathrm{as}^{2}(\gamma)r^2(\gamma))
\]
 Therefore, the relative error is
$
\var(S_n^2)/\var^2(\hat\ell_1)=\Theta(\ell_\mathrm{as}(\gamma)/r(\gamma))
$. By Lemma~\ref{lem3} there exists an $\alpha>1$ such that
\[\textstyle
\frac{r(\gamma)}{\ell_\mathrm{as}(\gamma)}=
\frac{r(\gamma)}{\ell_\mathrm{as}(\gamma^\alpha)}\times\frac{\ell_\mathrm{as}(\gamma^\alpha)}{\ell_\mathrm{as}(\gamma)}=o(1)\times\c O\left(\exp(-\frac{(\alpha^2-1)\ln^2(\gamma)}{2\sigma^2})\right),
\]
which shows that $\frac{\ell_\mathrm{as}(\gamma)}{r(\gamma)}$ grows at least at the exponential rate $\exp(\frac{(\alpha^2-1)\ln^2(\gamma)}{2\sigma^2})$.
\end{proof}

\subsection*{Proof of Lemma~\ref{lem}}
\begin{proof} First we show 1. To this end, recall that $\v X=\exp(\v Y)$, where $\v Y\sim\mathsf{N}(\v\nu,\Sigma)$. Further, recall the well-known property (which is strengthened in Lemma~\ref{lem3}) that for $i\not=j$ and
$\mathrm{Corr}(Y_i,Y_j)<1$, the pair $Y_i,Y_j$ is asymptotically independent in the sense that
\[
\bb P(Y_i>\gamma|Y_j>\gamma)=o(1),\qquad \gamma\uparrow\infty.
\]
In fact, Lemma~\ref{lem3} shows that this decay to zero is exponential.
The consequences of this are 
$
\bb P(\max_i Y_i>\gamma)\simeq \sum_i \bb P(Y_i>\gamma)
$
and
\[
\bb P(Y_k>\gamma,\max_{i\not =k} Y_i>\gamma)=o(\bb P(Y_k>\gamma))
\]
With these properties, we then have the lower bound:
\[
\begin{split}
\bb P(S>\gamma,X_k=M)&\geq \bb P(X_k=M>\gamma)\\
&\geq \bb P(X_k>\gamma,\max_{j\not=k}X_j<\gamma)\\
&=\bb P(Y_k>\ln\gamma,\max_{j\not=k}Y_j<\ln\gamma) \\
&=\bb P(Y_k>\ln\gamma)+o(\bb P(Y_k>\ln\gamma))\\
&=\bb P(X_k>\gamma)+o(\bb P(X_k>\gamma))
\end{split}
\]
Next, using the result $\bb P(S>\gamma,X_k=M<\ln\gamma)=o(\bb P(X_k>\ln\gamma))$ from Lemma~\ref{Lem 2}, we also have the analogous upper bound:
\[
\begin{split}
\bb P(S>\gamma,X_k=M)&= \bb P(X_k=M>\gamma)+\bb P(S>\gamma,X_k=M<\gamma)\\
&\leq \bb P(X_k>\gamma)+\bb P(S>\gamma,X_k=M<\gamma)\\
&= \bb P(X_k>\gamma)+o(\bb P(X_k>\gamma)),
\end{split}
\]
whence we conclude that $\bb P(S>\gamma,X_k=M)\simeq \bb P(X_k>\gamma)$. 

Next, we show point 2. Using the facts that: (1)
 the fewer the active constraints in any solution, the closer its minimum is to zero  (without constraints the minimum of \eqref{mu} is zero); (2) any solution   satisfies the \emph{Karush-Kuhn-Tucker} (KKT) necessary conditions:
\[
\begin{split}
\Sigma^{-1}\v\mu -\lambda_1\nabla g_1(\v\mu)-\lambda_2 \nabla g_2(\v\mu)
&=\v 0\\
\v\lambda\geq\v 0,\quad \v g(\v\mu)\geq \v 0,\quad\v\lambda^\top \v g(\v\mu) &=0,
\end{split}
\]
we can verify by direct substitution that $\v\mu^*$ satisfies the KKT conditions asymptotically as $\gamma\uparrow\infty$ and that it causes only one
constraint to be active ($g_1(\v\mu^*)=o(1)$). Moreover, it yields the asymptotic
 minimum: 
\[
\frac{1}{2}(\v\mu^*)^\top\Sigma^{-1}\v\mu^*=
\frac{(\ln(\gamma)-\nu_k)^2}{2\sigma_k^4}\v e_k^\top\Sigma\Sigma^{-1}\Sigma\v e_k=\frac{(\ln(\gamma)-\nu_k)^2}{2\sigma_k^2}
\]
Finally, we show point 3, which is the linchpin of the proposed methodology. To this end, consider the $(m+1)$-st moment with  $\v\mu\rightarrow\v\mu^*$ as $\gamma\uparrow\infty$:
\[
\begin{split}
\bb E_{\v\mu} \hat\hbar^{m+1}_k=\bb E_{\v 0} \hat\hbar^{m}_k&=\textstyle\bb E \exp\left(\frac{m\v\mu^\top\Sigma^{-1}\v\mu}{2}-m\v\mu^\top \Sigma^{-1}(\v Y-\v \nu)\right)\bb I\{S>\gamma, X_k=M\}\\
&=\textstyle \exp\left(\frac{(m^2+m)\v\mu^\top\Sigma^{-1}\v\mu}{2}\right)\bb P_{-m\v\mu}(S>\gamma, X_k=M)\\
&\simeq\textstyle \exp\left(\frac{(m^2+m)(\ln(\gamma)-\nu_k)^2}{2\sigma_k^2}\right)\bb P_{-m\v\mu^*}(S>\gamma, X_k=M)\\
\end{split}
\]
Next, notice that the measure $\bb P_{-m\v\mu^*}$ is equivalent to first simulating 
\[
Y_k\sim\mathsf{N}(\nu_k-m(\ln(\gamma)-\nu_k),\sigma_k^2)
\] and then, given $Y_k=y_k$, simulating all the rest of the components, denoted $\v Y_{-k}$, from the nominal Gaussian density $\phi_\Sigma(\v y-\v\nu)$ conditional on $Y_k=y_k$, that is, $\v Y_{-k}\sim
\phi_\Sigma(\v y-\v \nu|y_k)$. In other words, asymptotically, the effect of the change of measure induced by \eqref{mu} is to modify the marginal distribution of $X_k$ only.
Thus, repeating the same argument used to prove part 1, we have
\[
\bb P_{-m\v\mu^*}(S>\gamma, X_k=M)\simeq \bb P_{-m\v\mu^*}(Y_k>\ln\gamma)=\overline\Phi\left(\frac{(m+1)(\ln\gamma-\nu_k)}{\sigma_k}\right)
\]
Therefore, as $\gamma\uparrow\infty$, 
\begin{equation*}
\begin{split}
\bb E_{\v\mu} \hat\hbar^{m+1}_k&\textstyle\simeq  \exp\left(\frac{(m^2+m)(\ln(\gamma)-\nu_k)^2}{2\sigma_k^2}\right)\overline\Phi\left(\frac{(m+1)(\ln\gamma-\nu_k)}{\sigma_k}\right)\\
&=\textstyle\Theta\left( \frac{1}{\ln\gamma}\exp\left(-\frac{(m+1)(\ln(\gamma)-\nu_k)^2}{2\sigma_k^2}\right)\right)=\Theta(\ln^m(\gamma)\hbar_k^{m+1})
\end{split}
\end{equation*}
Then, the part 3 of Lemma~\ref{lem}  follows from putting $m=1$, and observing that
\[
\begin{split}
\frac{\var(\hat\hbar_k)}{\hbar_k^2}&=\frac{\bb E_{\v\mu}\hat\hbar_k^2}{[\bb P(S>\gamma,X_k=M)]^2}-1\simeq \frac{\bb E_{\v\mu}\hat\hbar_k^2}{[\bb P(X_k>\gamma)]^2}-1=\Theta(\ln(\gamma))
\end{split}
\]

\end{proof}


\begin{lem}\label{Lem 2}
We have 
$\bb P(S>\gamma,X_k=M<\gamma)=o(\bb P(X_k>\gamma))$
 as $\gamma\uparrow\infty$.
\end{lem}
\begin{proof} Let $\beta\in(0,1)$ and $M_{-k}=\max_{j\not=k}X_j$.
Then, using the facts: 
\[
\frac{\overline\Phi(\ln(\gamma-\gamma^\beta))}{\overline \Phi(\ln\gamma)}\simeq\exp\left(-\frac{\ln^2(\gamma-\gamma^\beta)-\ln^2(\gamma)}{2}\right)\frac{\gamma-\beta\gamma^\beta}{\gamma-\gamma^\beta}
\]
and 
\[
\ln^2(\gamma)-\ln^2(\gamma-\gamma^\beta)\simeq 2\frac{\ln(\gamma)}{\gamma^{1-\beta}}+o\left(\frac{\ln(\gamma)}{\gamma^{1-\beta}}\right),
\]
we obtain $\overline\Phi(\ln(\gamma-\gamma^\beta))\simeq \overline\Phi(\ln\gamma)$ for any $\beta\in(0,1)$. More generally, 
\[
\bb P(\ln(\gamma-\gamma^\beta)\leq Y_k\leq \ln\gamma)=o(\bb P(Y_k>\ln\gamma)).
\]
Then, we have $\bb P(S>\gamma,X_k=M<\gamma)=$
\[
\begin{split}
&=\textstyle\bb P(M_{-k}>\gamma^\beta,S>\gamma,X_k=M<\gamma)+\bb P(M_{-k}<\gamma^\beta,S>\gamma,X_k=M<\gamma)\\
&\leq \textstyle\bb P(\gamma^\beta<M_{-k}<X_k<\gamma)+\bb P(M_{-k}<\gamma^\beta,\;\gamma-(d-1)M_{-k}<X_k<\gamma)\\
&\leq\textstyle \bb P(\gamma^\beta<M_{-k},\;\gamma^\beta<X_k)+\bb P(\gamma-(d-1)\gamma^\beta<X_k<\gamma)
\end{split}
\]
Since for large enough $\gamma$ there exists a $\beta'\in(\beta,1)$ such that $(d-1)\gamma^\beta<\gamma^{\beta'}$, we have
 \[
\bb P(\gamma-(d-1)\gamma^\beta<X_k<\gamma)\leq \bb P(\gamma-\gamma^{\beta'}<X_k<\gamma)=o(\bb P(X_k>\gamma))
\]
The proof will then be complete if we can find a $\beta\in(0,1)$, such that ($u=\ln\gamma$)
\[
\bb P(M_{-k}>\gamma^\beta,X_k>\gamma^\beta)=
\bb P(\max_{j\not=k}Y_j>\beta u,Y_k>\beta u)=o(\bb P(Y_k>u))
\]
Since $\bb P(\max_{j\not=k}Y_j>\beta u,Y_k>\beta u)=\c O\left(\sum_{j\not=k}\bb P(Y_j>\beta u,Y_k>\beta u)\right)$, the last is equivalent to showing that the bivariate  normal probability 
$
\bb P(Y_j>\beta u,Y_k>\beta u)=o(\bb P(Y_k>u))
$
for some $\beta\in(0,1)$. This last part then follows from Lemma~\ref{lem3}, which completes the  proof.
\end{proof}

\begin{lem}[Gaussian Tail Probability]
\label{lem3} 
Let $Y_1\sim \mathsf{N}(\nu_1,\sigma_1^2)$ and 
$Y_2\sim \mathsf{N}(\nu_2,\sigma_2^2)$ be jointly
 bivariate normal with correlation coefficient $\rho\in(-1,1)$. Then, there exists an $\alpha>1$ such that
\[
\bb P(Y_1>\gamma,Y_2>\gamma)=o(\bb P(Y_1>\alpha\gamma) \wedge \bb P(Y_2>\alpha\gamma)),
\]
where $a\wedge b$ stands for $\min\{a,b\}$.
\end{lem}
\begin{proof} Without loss of generality, we may assume that $\sigma_1>\sigma_2$, so that
\[
\bb P(Y_1>\alpha\gamma) \wedge \bb P(Y_2>\alpha\gamma)\simeq\bb P(Y_2>\alpha\gamma)=\textstyle\Theta(\gamma^{-1}\exp(-\frac{(\alpha\gamma-\nu_2)^2}{2\sigma_2^2}))
\]
 Define the convex quadratic program:
\begin{equation}
\label{qpp}
\begin{split}
\min_{\v y}&\;\frac{1}{2}\v y^\top\Sigma^{-1}\v y\\
\textrm{subject to: }&\v y\geq \gamma\v 1-\v\nu,
\end{split}
\end{equation}
where $\Sigma_{11}=\sigma_1^2,\Sigma_{12}=\Sigma_{21}=\rho\sigma_1\sigma_2,\Sigma_{22}=\sigma_2^2$. Denote the solution as $\v y^*$. Then, we have the following asymptotic result \cite{hashorva2003multivariate}:
\[\textstyle
\bb P(Y_1>\gamma,Y_2>\gamma)=\Theta\left( \gamma^{-d_1}\exp\left(-\frac{(\v y^*)^\top\Sigma^{-1}\v y^*}{2}\right)\right),
\]
where $d_1\in\{1,2\}$ is the number of active constraints in \eqref{qpp}. 
Next, consider the quadratic programing problem which is the same as \eqref{qpp}, except that we  drop the first constraint (that is, we drop $y_1\geq\gamma-\nu_1$). The minimum of this second quadratic programing problem is  $\frac{(\gamma-\nu_2)^2}{2\sigma_2^2}$, and is achieved at the point
$
\tilde{\v y}=
\left((\gamma-\nu_1)\rho\sigma_2/\sigma_1,\gamma-\nu_2
\right)^\top
$. Note that since $\tilde y_1<\gamma-\nu_1$, we have dropped an active constraint.
 Since dropping an active constraint in a convex quadratic minimization achieves an even lower minimum, we have the strict inequality between the minima of the two quadratic minimization problems:
\[
0<\frac{(\gamma-\nu_2)^2}{2\sigma_2^2}<\frac{(\v y^*)^\top\Sigma^{-1}\v y^*}{2}
\] 
for any large enough $\gamma>\nu_2$. Hence, after rearrangement of the last inequality, we have
\[
\frac{\nu_2+\sigma_2\sqrt{(\v y^*)^\top\Sigma^{-1}\v y^*}}{\gamma}>1,
\]
and therefore there clearly exists an $\alpha$ in the range
\[
1<\alpha<\frac{\nu_2+\sigma_2\sqrt{(\v y^*)^\top\Sigma^{-1}\v y^*}}{\gamma}
\]
For such an $\alpha$ (in the above range), we have
\[
\frac{(\alpha\gamma-\nu_2)^2}{2\sigma_2^2}<\frac{(\v y^*)^\top\Sigma^{-1}\v y^*}{2}
\]
Therefore,
$
\exp(-\frac{(\v y^*)^\top\Sigma^{-1}\v y^*}{2})=o\left(\exp(-\frac{(\alpha\gamma-\nu_2)^2}{2\sigma_2^2})\right)
,\; \gamma\uparrow\infty$,
and the exponential rate of decay
of $\bb P(Y_1>\gamma,Y_2>\gamma)$ is greater than  that of $\bb P(Y_2>\alpha\gamma)$.
This completes the proof.
\end{proof}
%
\section*{Acknowledgements}
Zdravko Botev has been supported by the Australian Research Council grant DE140100993. Robert Salomone has been supported by the Australian Research Council Centre of Excellence for
Mathematical \& Statistical Frontiers (ACEMS), under grant number CE140100049.

\bibliographystyle{plain}
\bibliography{lognorm}

\end{document}